\documentclass[a4paper,onecolumn,11pt,accepted=2025-04-13]{quantumarticle}
\pdfoutput=1
\usepackage[utf8]{inputenc}
\usepackage[english]{babel}
\usepackage[T1]{fontenc}
\usepackage{amsmath}
\usepackage{amssymb}
\usepackage{multirow}
\usepackage{braket}
\usepackage{ dsfont }
\usepackage{amsthm}
\usepackage{soul}
\usepackage{tikz}
\usepackage{lipsum}
\usepackage[numbers,sort&compress]{natbib}
\usepackage{hyperref}

\newtheorem{theorem}{Theorem}
\newtheorem{lemma}{Lemma}

\newtheorem{definition}{Definition}

\newtheorem{innercustomthm}{Theorem}
\newenvironment{mythm}[1]
  {\renewcommand\theinnercustomthm{#1}\innercustomthm}
  {\endinnercustomthm}

\begin{document}

\title{Pauli path simulations of noisy quantum circuits beyond average-case}

\author{Guillermo González-García}
\affiliation{Max-Planck-Institut für Quantenoptik, Hans-Kopfermann-Str.~1, 85748 Garching, Germany}
\affiliation{Munich Center for Quantum Science and Technology (MCQST), Schellingstr. 4, D-80799 Munich, Germany}
\email{guillermo.gonzalez@mpq.mpg.de}
\author{J. Ignacio Cirac}
\affiliation{Max-Planck-Institut für Quantenoptik, Hans-Kopfermann-Str.~1, 85748 Garching, Germany}
\affiliation{Munich Center for Quantum Science and Technology (MCQST), Schellingstr. 4, D-80799 Munich, Germany}
\author{Rahul Trivedi}
\affiliation{Max-Planck-Institut für Quantenoptik, Hans-Kopfermann-Str.~1, 85748 Garching, Germany}
\affiliation{Deparment of Electrical and Computer Engineering, University of Washington, Seattle, Washington 98195, USA}
\maketitle

\begin{abstract}

For random quantum circuits on $n$ qubits of depth $\Theta(\log n)$ with depolarizing noise, the task of sampling from the output state can be efficiently performed classically using a Pauli path method \cite{Aharonov2023_paulipaths}. This paper aims to study the performance of this method beyond random circuits. We first consider the classical simulation of local observables in circuits composed of Clifford and $T$ gates --- going beyond the average-case analysis, we derive sufficient conditions for simulability in terms of the noise rate and the fraction of gates that are $T$ gates, and show that if noise is introduced at a faster rate than $T$ gates, the simulation becomes classically easy. As an application of this result, we study 2D QAOA circuits that attempt to find low-energy states of classical Ising models on general graphs. There, our results shows that for hard instances of the problem, which correspond to Ising model's graph being geometrically non-local, a QAOA algorithm mapped to a geometrically local circuit architecture using SWAP gates does not have any asymptotic advantage over classical algorithms if depolarized at a constant rate. Finally, we illustrate instances where the Pauli path method fails to give the correct result, and also initiate a study of the trade-off between fragility to noise and classical complexity of simulating a given quantum circuit. 
\end{abstract}

\tableofcontents

\section{Introduction}

Fault-tolerant quantum computers are believed to offer substantial speed-ups in different problems, such as quantum simulation \cite{georgescu2014_quantumsimulation} or prime factorization \cite{Shor_1994,Shor_1997}. While there has been impressive experimental progress in error correction \cite{ibm2024qec,google2023qec}, current quantum devices are unable to reach the fault-tolerant regime, due to the high levels of noise and limited number of qubits. These are the so-called NISQ (Noisy Intermediate Scale Quantum) devices \cite{Preskill2018NISQ}. As a consequence, there is now widespread interest in exploring the capabilities of NISQ devices \cite{chen2023complexity}. One family of widely studied methods are variational quantum algorithms \cite{cerezo2021variational}, which attempt to combine classical and quantum resources to solve complex optimization problems. Out of these, QAOA (for solving classical combinatiorial optimizazion problems) \cite{farhi2014qaoa}, and VQE (for preparing groundstates) \cite{peruzzo2014variational} stand out. However, variational algorithms are not without problems. A number of challenges have been identified, such as barren plateaus \cite{cerezo2021barrenplateaus}, expressibility of the ansatz \cite{expressibility1_2021,expressibility2_2021_nakaji}, or reachability deficits \cite{reachabilitydeficits_QAOA_2020}. Furthermore, it has been noted that the presence of noise imposes severe limitations \cite{stilck2021limitations,stilck2023_limitations_transport,gonzalez2022_errors} on these algorithms. 

Following the recent quantum advantage experiments \cite{google2019quantumsupremacy,quantum-advantage-pan,kim_ibm2023evidence}, there has been widespread interest in the classical simulation of both noisy and noiseless quantum circuits \cite{Aharonov2023_paulipaths,tindall2023efficient,kechedzhi2024effective,shao2023simulating,liao2023simulation,fontana2023classical_LOWESA,rudolph2023classical_LOWESA,anand2023classical,Nemkov_2023,gao2018efficient_simulation,rajakumar2024_IQP,tanggara2024Paulipath}. Pauli path based algorithms are a family of algorithms that have been proposed to perform this task. They work by performing a Feynman path expansion in the Pauli basis (also known as Pauli path integral), and then truncating the Pauli operators with higher Hamming weight \cite{Aharonov2023_paulipaths,shao2023simulating,fontana2023classical_LOWESA,rudolph2023classical_LOWESA,Nemkov_2023,tanggara2024Paulipath}, thus keeping track of the local dynamics in the Pauli basis. This approach has been employed, for example, to simulate hydrodynamics efficiently \cite{vonKeyserlingk_2022_operator_backflow,rakovzky_2022_dissipation-assisted}. Aharonov et. al \cite{Aharonov2023_paulipaths} first used this strategy to successfully simulate sampling from depolarized random circuits under two assumptions --- gate-set orthogonality and anti-concentration \cite{Dalzell2022_anti-concentration}, both of which were satisfied for random circuits of depth $d=\Omega(\log n)$ \cite{Dalzell2022_anti-concentration}. Importantly, the models of random circuits that have been theoretically proposed for testing quantum advantage for a sampling task require anti-concentration \cite{bouland2019complexity}.  Consequently, this line of work identified stringent constraints on asymptotic quantum advantage in random circuit sampling task in the presence of depolarizing noise. Instead of sampling, one could also study the simpler problem of computing expectation values of local observables.  Performing an average-case analysis, the Pauli paths method can be shown to succeed with high probability for circuits sampled from a random-circuit ensemble that satisfies gate-set orthogonality but not anti-concentration. Recently, this has been exploited to practically simulate noisy variational circuits. In particular, in \cite{shao2023simulating,fontana2023classical_LOWESA} it is shown that the method succeeds to compute noisy expectation values with high probability for any circuit depth and constant error rate, when the optimization parameters are sampled randomly and independently.

While random circuits are a good model for theoretically understanding the efficacy of classical algorithms and are also relevant for current experiments \cite{google2019quantumsupremacy,quantum-advantage-pan,ZHU2022_advantagerandom}, several practically relevant computational tasks require circuits that are not random. For instance, many-body dynamics governed by time-independent Hamiltonian would in general be digitally simulated by repeating the same layer of unitaries again and again \cite{berry2015_Hsimulation,childs2021_trotter}. Such circuits become even less like random circuits if the translational invariance \cite{farrell2024_simu,barratt2021parallel,Mansuroglu_2023_variational_invariance} or special symmetries \cite{Maruyoshi_2023_conserved,vanicat2018_trotter_conserved} of the target many-body Hamiltonian are taken into account. Even  ground state problems for many-body problems arising in physics are often simulated variationally using translationally invariant circuits, which would again not satisfy the typically studied template of random circuits \cite{farrell2024_gs,Babbush_2018qsim,Keyserlingk2018_OTOC}. Furthermore, the response of typical circuits to noise might not always be representative of the behavior of individual instances of noisy circuits of practical interest. For example, important differences with regards to the spreading of noise have been identified. While typical circuits exhibit a fast propagation of errors \cite{gonzalez2022_errors,debmishra2024bounds} that is harmful to the computation, the presence of conservation laws in specific algorithms can suppress this proliferation of errors \cite{schiffer2024proliferation}. A better understanding of the performance of classical algorithms, such as that of the Pauli-path method, for noisy and non-random circuits would help identifying where a quantum advantage could be expected in practically relevant quantum simulation problems. 



In this paper we further the theoretical understanding of the  Pauli path method to compute expectation values of depolarized quantum circuits, and focus on going beyond the average-case analysis of random circuits. We first study 2D circuits that are comprised of Clifford and $T$ gates in section \ref{section:T_gates}. We find that for this class of circuits there is a threshold error rate above which the simulation becomes easy. For circuits in which $T$ gates are uniformly distributed throughout the circuit, this threshold error rate depends on the fraction of $T$ gates present in the circuit. Importantly, this provides a sufficient condition for simulability. This contrasts with the average-case analysis, which only guarantees that the algorithm succeeds with high probability for randomly drawn circuits. As an application of this result, in section \ref{section:Ising_hamiltonian} we study 2D QAOA circuits that attempt to find low-energy states of classical Ising models on general graphs. There, our results shows that for hard instances of the problem which correspond to Ising model's graph being geometrically non-local, a QAOA algorithm mapped to a geometrically local circuit architecture using SWAP gates does not have any asymptotic advantage over classical algorithms if depolarized at a constant rate.

In section \ref{Section: counterexample}, we show by explicit construction that the Pauli-path method fails in the worst-case for all noise rates $p < 1-\sqrt{2/3}$. This has been hinted in different works \cite{gao2018efficient_simulation,rudolph2023classical_LOWESA,fontana2023classical_LOWESA} and it is expected that there exists a small subset of (possibly very structured) circuits where classical simulation fails. The circuit that we construct is remarkably simple, and can be trivially simulated using other methods. It is therefore not an example of a generally classically intractable circuit, but a simple example of the worst-case failure of Pauli path based algorithms.

\section{Summary of results}

We study the problem of computing expectation values of observables in the presence of local depolarizing noise of strength $p$, using the Pauli paths method. 
The algorithm follows the prescription in \cite{Aharonov2023_paulipaths,shao2023simulating,fontana2023classical_LOWESA}: the Pauli observable is evolved in the Heisenberg picture, and it is expressed in the Pauli basis after each layer of unitaries. This is sketched in Fig. \ref{fig:Pauli_paths}. As shown in \cite{Aharonov2023_paulipaths}, the contribution of the Pauli paths of weight $\ell=O(\log n)$ \footnote{Throughout the paper we will use the standard asymptotic notation from complexity theory that is detailed in appendix \ref{appendix:notation}.} can be computed efficiently with a classical algorithm. The high weight Pauli paths are discarded, since their contribution is damped by depolarizing noise.

\begin{figure*}
    \centering
    \includegraphics[scale=0.7]{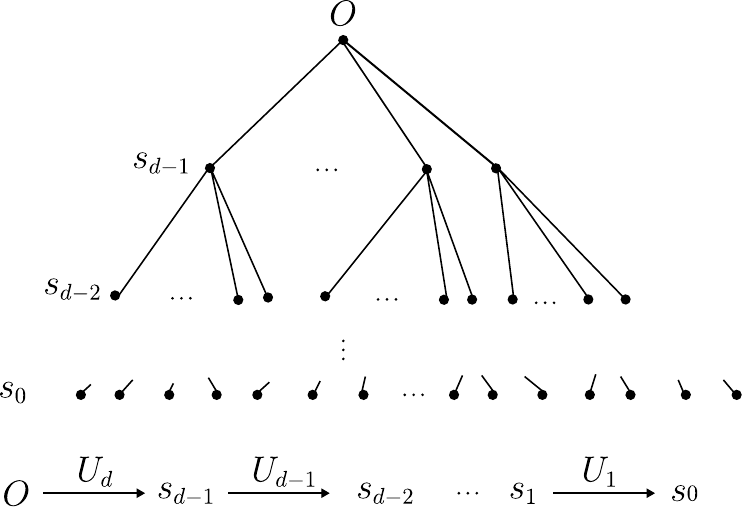}
    \caption{Graphical representation of the Pauli path method. The operator $O$ is evolved under Heisenberg evolution, expressing it in the Pauli basis after each layer. The process can be depicted as a tree: first $U_d$ is applied, and $U_d^{\dagger}OU_d$ can be written as a linear combination of Pauli strings, denoted as $s_{d-1}$ in the figure. Then, each of the $s_{d-1}$ will be evolved according to $U_{d-1}$, which will yield a superposition of different $s_{d-2}$. The process is then repeated until all the layers have been applied. We refer to a single branch of this tree as a Pauli path $s$, and it corresponds to a specific configuration of Pauli strings, $s=(s_d,s_{d-1},..,s_0)$. The sum over all possible Pauli paths yields the exact Heisenberg evolution of the operator $O$.}
    \label{fig:Pauli_paths}
\end{figure*}

\subsection{Prior work}

We are considering the problem of classically approximating expectation values of noisy quantum circuits in the presence of single qubit depolarizing noise. Due to an increase in the entropy of qubits due to the depolarizing noise, this task is trivial for circuit of depth $d=\omega(\log n)$, since in this the full quantum state of the qubits will be $1/n^{\omega(1)}-$close to the maximally mixed state $I/2^n$ \cite{aharonov1996limitations}. However, different techniques are needed to address circuits with depth $d=O(\log n)$. One of such techniques are the Pauli path based methods, which only keep track of the local dynamics in the Pauli basis.

Aharonov et al. \cite{Aharonov2023_paulipaths} first rigorously analyzed a Pauli path based polynomial time algorithm for the task of simulating noisy Random Circuit Sampling (RCS). In their work, they showed the enumeration of all Pauli paths up to a Pauli weight $\ell=O(\log n)$ can be done efficiently for circuits of depth $\Theta(\log n)$, with runtime $2^{O(\ell)}$. In their analysis they invoked two properties of random circuits: anti-concentration \cite{Dalzell2022_anti-concentration} and gate-set orthogonality. When both properties are present, they were able to provide a polynomial time algorithm that samples from the distribution of noisy quantum circuits within a small total variation distance. Specifically, computing all the paths up to a weight $\ell=O(\log n)$ is sufficient to provide an $O(1/\mathrm{poly}(n))$ approximation of the output probability distribution.

For the task of computing noisy expectation values of Pauli observables the situation is simpler, since anti-concentration is no longer required. In fact, as shown in \cite{shao2023simulating,fontana2023classical_LOWESA,rudolph2023classical_LOWESA}, the enumeration of the Pauli paths of weight $O(\log n)$ is efficient in all cases, regardless of the circuit depth or geometry. However, some randomness is still needed to control the error in the algorithm. For example, it is known that if a circuit is sampled from a distribution that satisfies gate-set orthogonality, without requiring anti-concentration, the method succeeds with high probability. In particular, the analysis in \cite{shao2023simulating,fontana2023classical_LOWESA} focuses on noisy variational circuits. They show that the classical method can successfully approximate expectation values for any constant noise rate and any circuit depth when averaging over the parameters of the circuits. That is, they perform an average-case analysis, and show that the algorithm succeeds for typical circuits with a fixed structure. These results are inline with previous works \cite{gao2018efficient_simulation}, which show that classical simulation of noisy circuits with a constant noise rate is efficient for almost all circuits, with the exception of a small subset of structured circuits. Finally, in \cite{mele2024noiseinduced} this is shown for more general noise models, providing a classical method to approximate quantum expectation values that succeeds for most circuits even in the presence of non-unital noise. This contrasts with the RCS task, since noisy random circuits under non-unital noise do not anti-concentrate, and therefore conventional techniques break down when it comes to sampling \cite{fefferman2023_effectnonunitalnoiserandom}.

The classical simulation of noiseless Clifford circuits is known to be efficient, following the celebrated Gottesman-Knill theorem \cite{gottesman1998knill_theorem,aaronson2004clifford}. The addition of $T$ gates to a Clifford gate-set yields a universal gate-set, and is not expected to be efficiently classically simulable. However, there are classical algorithms that require a run-time only exponential in the total number of $T$ gates (as opposed to the circuit depth or the number of qubits) to efficiently sample from Clifford+$T$ gates circuits \cite{bravyi2016_simulationTgates}. Furthermore, there exist several works addressing classical simulability of such circuits with noisy $T$ gates even when the Clifford gates are noiseless \cite{buhrman2006simulationTgates,rall2019simulation_T_gates_Pauli}. Precisely, Ref. \cite{buhrman2006simulationTgates} showed that a $T$ gate followed by sufficiently strong depolarizing noise ($p \gtrsim 0.45$) can be written as convex mixture of Clifford operations, which enables efficient classical simulation in Clifford + $T$ gates circuits, when the $T$ gates suffer from strong depolarizing noise. Similarly, Ref. \cite{rall2019simulation_T_gates_Pauli} analyzed Pauli back-propagation techniques, and showed efficient simulation for $T$ gates followed by depolarizing noise of rate $p >1-1/\sqrt{2} \simeq0.3$.

\subsection{Main results}

\subsubsection*{Classical simulation of Clifford + $T$ gates circuits}

In section \ref{section:T_gates} we consider 2D circuits that are comprised of Clifford + $T$ gates. The addition of $T$ gates to the Clifford gate-set yields a universal gate-set \cite{nielsen2010quantum} that is often available in digital quantum computers \cite{murali2019nativegates}. Furthermore, we focus on circuits with a 2D topology since they are often used in experimental settings \cite{google2019quantumsupremacy,kim_ibm2023evidence}, although our analysis generalizes straightforwardly to circuits in higher dimensions. For Clifford circuits, the algorithm trivially succeeds since there is only one nonzero Pauli path, because Clifford unitaries map Pauli strings to Pauli strings. This is not surprising, since Clifford circuits are efficiently simulable classically \cite{aaronson2004clifford}.  In the absence of noise, circuits made out of Clifford + $T$ gates can be classically simulated in time that is exponential in the number of $T$ gates.

\begin{figure*}[h]
    \centering
    \includegraphics[scale=0.65]{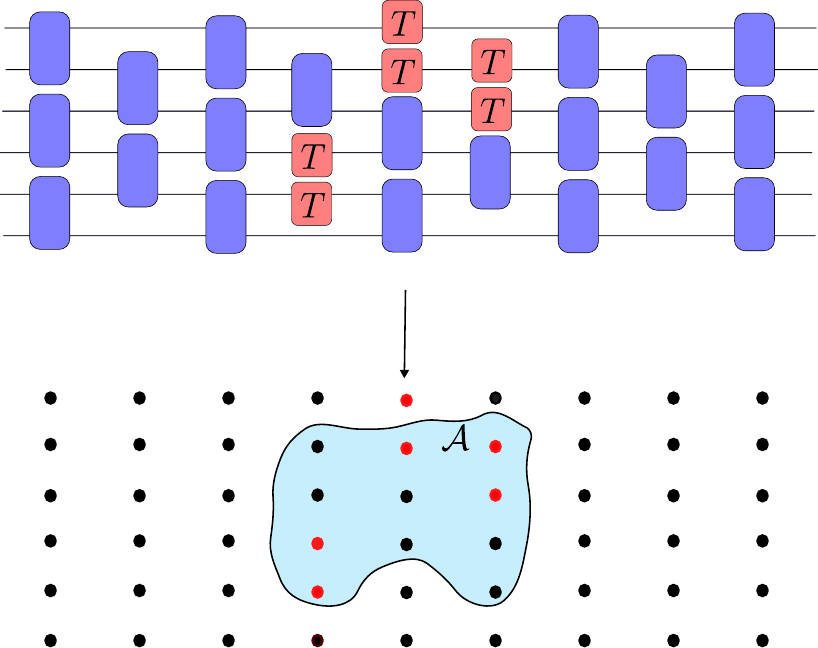}
    \caption{Sketch of the geometrical condition on the distribution of $T$ gates that we require for simulatibility. In the figure we represent a 1D circuit instead of a 2D one, for simplicity. In the upper half of the figure a circuit is depicted. The blue gates are Clifford gates, while the red gates represent $T$ gates. In the lower half of the figure the same circuit is represented, with the qubits represented as dots, while the gates are not depicted. A black dot indicates that a Clifford gate is applied to the qubit, and a red dot indicates that a $T$ gate is applied to the qubit. A subset $\mathcal{A}$ has been highlighted, containing a fraction of $T$ gates of $5/11$. For simulability, we require that all subsets of contiguous gates of size greater than $k \log n$ contain a fraction of $T$ gates no greater than $Q$, for some $k=\Theta(1)$. This is defined formally in Definition \ref{definition:sparseness}.}
\label{fig:sparsity}
\end{figure*}

We apply the Pauli paths method in this setting, and provide sufficient conditions for the simulability of such circuits.
Specifically, we find that for Clifford + $T$ gates circuits there is a threshold error rate above which the simulation becomes easy, which depends on the fraction of $T$ gates and their position in the circuit. To show this, we require that the $T$ gates are distributed sufficiently uniformly throughout the circuit. The precise geometric condition, sketched in Fig.~\ref{fig:sparsity}, is that for all set of $k$ contiguous gates, with $k>a \log n$ for some constant $a=\Theta(1)$, the fraction of $T$ gates in the set is no larger than a constant $Q$. 

\begin{mythm}{1}[T gates, informal]
Consider a 2D quantum circuit comprised of Clifford and $T$ gates such that for any set of contiguous set of gates larger than $a\log n$ for $a=\Theta(1)$, the fraction of $T$ gates contained in the set is no larger than $Q$. Then for any error rate $p>1-1/2^Q=\Theta(Q)$ as $Q \to 0$, the noisy expectation value of an observable $O$ given by a linear combination of at most $\mathrm{poly}(n)$ Pauli strings can be classically computed to precision $\epsilon \|O\|$ in $O(\mathrm{poly}(n,1/\epsilon))$ time.
\end{mythm}

This result is achieved by refining the enumeration of nonzero Pauli paths present in \cite{Aharonov2023_paulipaths,fontana2023classical_LOWESA,shao2023simulating} by using the fact Clifford gates map Pauli strings to Pauli strings and therefore cannot increase the number of nonzero Pauli paths. \noindent Note that, in the small $Q$ limit, the threshold condition reduces to $p>Q \log2$, which highlights the intuition behind this result: the classical simulation becomes easy if the rate at which errors are being introduced is greater than the rate at which quantum magic is introduced by the $T$ gates, since the Clifford gates introduce noise but no quantum magic. Hence, this results illustrates the ``competition'' between  quantum magic and noise. We remark that this ``competition'' has already appeared in previous works \cite{buhrman2006simulationTgates,rall2019simulation_T_gates_Pauli}, which already showed efficient classical simulability of noisy Clifford + $T$ gates circuits even with noiseless Clifford gates, provided that the $T$ gates are followed by depolarizing noise higher than a constant threshold. In contrast, our results indicate that it might be efficient to classically compute observables in circuits with an $o(1)$ fraction of $T-$gates at \emph{any rate of depolarizing noise}. 

We stress that we find the existence of a noise threshold above which classical simulation is efficient, which contrasts with previous average-case Pauli path analysis \cite{Aharonov2023_paulipaths,gao2018efficient_simulation,shao2023simulating,fontana2023classical_LOWESA,rudolph2023classical_LOWESA}, which manage to show efficient simulability on average cases for $\emph{any}$ constant noise rate. This highlights the fact that, while typical noisy circuits are provably efficiently simulable, there exists a small subset of structured circuits that are not expected to be simulable, for sufficiently small noise rates. Qualitatively, this is consistent with the threshold theorem \cite{aharonov1999faulttolerantquantum} from quantum error correction which indicates that below the fault-tolerance threshold, noisy circuits can be designed to simulate noiseless circuits and hence excludes the possibility of a classical algorithm to simulate \emph{all} circuit at \emph{any} noise rates. However, whether this qualitative expectation can be made rigorous remains unclear --- we remark that for the threshold theorem to apply, we need the ability to introduce fresh ancillas in $\ket{0}$ state in the middle of the circuit (or equivalently implement a restart operation). This is not included in the circuit family that we consider since we assume all the qubits involved in the circuit to experience depolarizing noise at every time-step from the start of the circuit. Finally, we also note that it is straightforward to generalize this result to circuits that contain non-Clifford gates other than $T$ gates, with only some constant prefactors changing.

\subsubsection*{Variational circuits}

As an application of the result above, in section \ref{section:Ising_hamiltonian} we study 2D variational circuits \cite{QAOA_review_2024} that solve classical combinatorial optimization problems. We specifically study circuits that attempt to find the ground energy of Ising Hamiltonians defined on a graph $G = (V, E)$ of bounded degree, with each vertex being an independent spin:
\[
H=\sum_{(i,j) \in E}J_{ij}Z_iZ_j+\sum_i b_i Z_i
\]
and where the couplings $J_{i, j}$ are $O(1)$ non-zero constants. The family of variational circuits that we study is sketched in Fig. \ref{fig:variational_ansatz} --- it comprises of $m$ layers of unitaries $L_1, L_2 \dots L_m$ where
\[
L_j = e^{-i\gamma_j H} e^{-i\alpha_j H_M},
\]
where $e^{-i\alpha_i H_M}$ is an $O(1)$-depth unitary that introduces entanglement in the qubits \cite{linghua2022adaptive} (e.g. in QAOA, $H_M = X_1 + X_2 + \dots X_n$) and $\alpha_i, \gamma_i$ are parameters of the variational circuits. Importantly, we assume that to implement $e^{-i \gamma_i H}$ on a 2D circuit geometry, we introduce layers of SWAP gates whenever $H$ itself is not 2D geometrically local. The geometry of the graph $G$ on which the Hamiltonian $H$ is defined plays an important role in our analysis: if the graph is planar (or almost planar), there are known classical efficient algorithms that can approximate the ground energy \cite{bansal2008classical_Ising}. On the other hand, for hard instances of the problem, the geometry of the graph is non-planar, and hence typically a large number of SWAP gates are needed to embed the graph in a 2D architecture. This places us in the regime were the circuit is dominated by Clifford gates, and according to our results above the output of the noisy can be approximated classically in polynomial time, for any constant noise rate:

\begin{figure*}[h]
    \centering
    \includegraphics[scale=0.85]{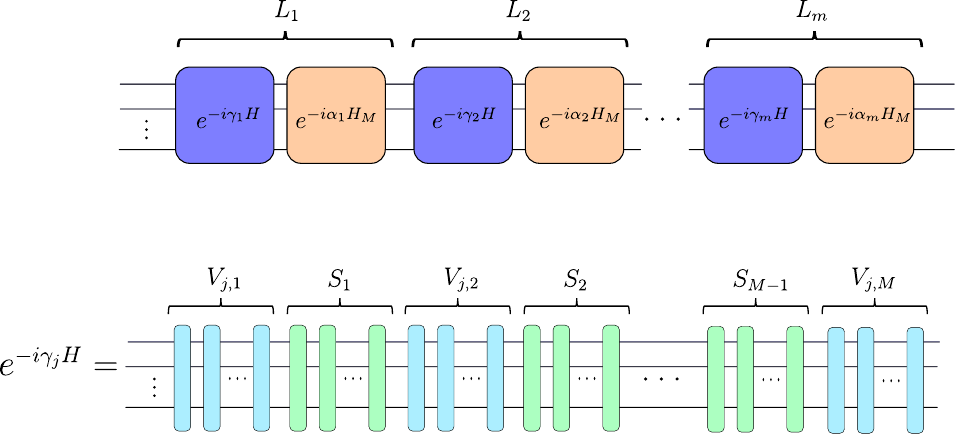}
    \caption{Graphical representation of the construction of the 2D variational circuits that we consider. The circuit is constructed by applying $m$ variational layers of the form $L_j=e^{-i \gamma_j H}e^{-i \alpha_j H_M}$, where $H$ is the objective Hamiltonian, $H_M$ is the mixer Hamiltonian, and $\gamma_j,\alpha_j$ are optimization parameters. The only assumption on $H_M$ is that the term $e^{-i \alpha_j H_M}$ can be implemented by a constant depth circuit. Furthermore, $e^{-i \gamma_j H}$ is implemented by applying computing layers $(V_{j,1} \dots V_{j,M})$, depicted in blue, and permuting layers $(S_{1} \dots S_{M-1})$, depicted in green. The computing layers contain Clifford and $T$ gates and implement the terms of the Hamiltonian $H$, $e^{-iZ_{i_1}Z_{i_2}}$ and $e^{-iZ_{i}}$, while the permuting layers allow for the embedding of the Hamiltonian graph into the 2D circuit architecture and contain only SWAP gates. This construction is very general and allows for the implementation of any bounded degree Ising Hamiltonian of the form of Eq. \ref{eq:ising_hamiltonian} with only a constant depth for the computing unitaries, $\sum_k  \mathrm{depth}(V_{j,k})=O(1) \forall j$. The depth of the permuting unitaries will depend on the geometry of the graph that needs to be embedded. }
    \label{fig:variational_ansatz}
\end{figure*}

\begin{mythm}{2}[Variational circuits, informal]
Consider a classical Ising Hamiltonian $H$ on a graph of constant degree, and a noisy variational circuit $\mathcal{C}$ for finding the ground energy of $H$ which, as described above, uses SWAP gates to implement a non-local $H$ with a geometrically local architecture. Then, there is either (a) a classical algorithm that approximates the ground energy of $H$ up to precision $\epsilon n$ in $O(n 2^{1/\epsilon^2})$ time, or (b) a classical algorithm that approximates the output energy of the noisy implementation of the circuit $\mathcal{C}$ up to precision $\epsilon n$ in $\mathrm{poly}(n,1/\epsilon)$ time for any constant noise rate $p$.
 \end{mythm}

We emphasize that this result relies on a specific embedding of the target Hamiltonian $H$ on the available geometrically local circuit architecture using SWAP gates, which is reasonable to assume in most applications. It does not apply to variational circuits that are not of this format i.e.~variational circuits that do not implement $\exp(-i\gamma_j H)$ in their ansatz. Furthermore, while we have assumed for simplicity that the only non-Clifford gates available are $T$ gates, the result can easily be extended to any other gate-set.

\subsubsection*{Worst-case performance}

In section \ref{Section: counterexample} we explore the worst-case limitations of the algorithm.  We denote by $\left<O\right>_{\ell}$ the output of the classical algorithm, which contains the contribution of all paths up to weight $\ell=\Theta(\log n)$, and by $\left<O\right>_{\mathrm{err}}$ the expected observable at the output of the noisy quantum circuit. We show that that in the most general case computing that Pauli paths with weight up to $\ell=\Omega(\log n)$ is not enough: $\ell$ must grow at least superlogarithmically i.e.~$\ell=\omega(\log n)$ to ensure that the error $\mathrm{Error}=\left|\left<O\right>_{\ell} - \left<O\right>_{\mathrm{err}} \right|$ does not grow exponentially with $\ell$. This implies that the classical algorithm needs a superpolynomial running time, $T=\omega(\mathrm{poly}(n))$:

\begin{mythm}{3}[No-go result, informal]
Consider the Pauli paths method to compute the expectation value of a convex combination of $g(n)=O(\mathrm{poly}(n))$Pauli observables $O=\sum_{k=1}^{g(n)} a_k O_k$, applied with a cut-off $\ell=\Omega(\log n)$. To compute the noisy expectation value $\left<O\right>_{\mathrm{err}}$ to any given precision, for all error rates and for any $O(\log n)$ depth circuit, it is necessary that the Pauli cut-off scales as $\ell=\omega(\log n)$.
\end{mythm}

To show this, we explicitly construct a circuit $\mathcal{C}$ for which naive application of the Pauli paths algorithm yields an error that grows exponentially with the Pauli weight cut-off $\ell$, $\mathrm{Error}=2^{\Omega(\ell)}/\left(\log n\right)^2=\Omega(\mathrm{poly}(n))/\left(\log n\right)^2$, when $\ell=\Theta(\log n)$, $\forall p < 1-\sqrt{2/3}$. We remark that the circuit is trivially easy to classically simulate, since it contains a single layer of unitaries. Furthermore, one could simply apply the algorithm to approximate each of the $O_k$ individually, which would trivially work. However, we argue that this simple example illustrates the limitations of the method: in the worst-case there can be an exponential growth of the error as a consequence of the truncation. This result is consistent with previous work suggesting that generic (random) noisy circuits are classically easy to simulate, but there is a small subset of circuits that cannot be simulated efficiently classically \cite{gao2018efficient_simulation}.

\subsubsection*{Root analysis}
Finally, in appendix \ref{section:root_analysis} we initiate a study of the trade-off between classical complexity and sensitivity to noise. To do this, we make use of the fact that, following the Pauli paths method, the noisy expectation value $\left<O\right>_{\mathrm{err}}(p)$ is written as a polynomial in the error rate $p$ (the precise construction of the polynomial is described in section \ref{section:setting}).
We then use the analytic theory of polynomials \cite{rahman2002analytic} to derive a relation between the roots of the polynomial $\left<O\right>_{\mathrm{err}}(p)$ and its behavior. Specifically, we find that if all the roots are real there is a trade-off between simulability and sensitivity to noise: $\left<O\right>_{\mathrm{err}}(p)$ can either be approximated efficiently classically, or is highly fragile to noise (for small noise rates).

Precisely, we consider Clifford + $T$ gates circuits, and a Pauli observable $O$, whose noisy expectation value $\left<O\right>_{\mathrm{err}}(p)$ can be expressed as a polynomial in  $p$ with all roots real. We further assume that $O$ is not vanishingly small in the noiseless setting, $\left<O\right>_{\mathrm{ideal}}=\Omega(1/\mathrm{poly}(n))$. Then, if there is no classical algorithm that can approximate the expectation value $\left<O\right>_{\mathrm{err}}(p)$ up to precision $\epsilon$ in time $T=O(\mathrm{poly}(n,1/\epsilon))$ $\forall p \in (0,1]$, we find the noisy expectation value must be vanishingly small for small constant noise rates:
\[
\left|\left<O\right>_{\mathrm{err}}(p)\right|=\frac{1}{\omega(\mathrm{poly}(n))}, \forall p \in \left(0,p^*\right],
\]
for some $p^{*}=\Theta(1)$, indicating that the circuit is highly fragile to noise.

We remark that the assumption that all the roots are real is not general, as one can provide circuits for which this does not hold. Furthermore, since computing the roots is as hard as computing the expectation value, it is not easy to characterize the subset of circuits with real roots. However, it is trivial to provide circuits that fulfill this condition: for example, all Clifford circuits will have all roots real. We pose as an open question whether a meaningful relation between the roots of the polynomial and the properties of the quantum circuit can be derived.

 \section{Setting and notation}\label{section:setting}

We consider a quantum circuit $\mathcal{C}$ consisting of $d$ layers of $D$-qubit (with $D=\Theta(1)$) gates $\mathcal{U}_i$, with an initial product state $\rho_0$. Mathematically, the noisy circuit is given by the quantum channel $\Phi_{\mathcal{C},p}^{\mathrm{noisy}} = \left[\bigcirc_{i=1}^d T_p \circ \mathcal{U}_i\right] \circ T_p$, where layers of noise $T_p$ are applied after every layer of unitaries, and also before the first layer. Specifically, each layer $T_p$ consists of applying single-qubit depolarizing noise of strength $p$ to every qubit, $T_p=\mathcal{N}_{p}^{\otimes n}$, with $\mathcal{N}_{p}(\rho)=(1-p)\rho + p I/2$.

We are interested in classically approximating the expectation value of a noisy observable $O$ up to absolute error $\epsilon \|O\|$. $O$ is given by a linear combination of at most $g(n)=O(\mathrm{poly}(n))$ Pauli observables, $O=\sum_{k=1}^{g(n)} a_k O_k$, where the coefficients $a_k$ are known. We will denote the noisy expectation value of $O$ by $\left<O\right>_\mathrm{err}=\mathrm{tr}(O\Phi_{\mathcal{C},p}^{\mathrm{noisy}}(\rho_0))$. In this expression, $\rho_0$ is the initial state, which is taken to be a product state.

The Pauli paths approach to solve the above problem follows the method in \cite{Aharonov2023_paulipaths}, and evolves the operator using in the Heisenberg evolution, expanding in the Pauli basis after each layer, and keeping track of the coefficients (this is represented in Fig. \ref{fig:Pauli_paths}). We denote a Pauli string by $s_i \in \{I,X,Y,Z\}^{\otimes n}$. A Pauli path is a sequence of $d+1$ Pauli strings, $s=(s_0,s_1, ... ,s_d)$. As done in \cite{Aharonov2023_paulipaths}, one can then express

\[
\left<O\right>_{\mathrm{err}}=\sum_s f(s)(1-p)^{|s|},
\]

\noindent with 

\begin{align}\label{eq:expression_Pauli_path}
f(s)=\frac{1}{2^{n(d+1)}}\mathrm{tr}\left(O s_d\right) \mathrm{tr}\left(s_d U_d s_{d-1} U_d^{\dagger} \right) ...  \mathrm{tr} \left( s_1 U_1 s_{0} U_1^{\dagger} \right) \mathrm{tr}\left(s_0 \rho_0 \right).
\end{align}
The term $\mathrm{tr}\left(s_k U_k s_{k-1} U_k^{\dagger} \right)$ can be understood as the transition amplitude of the circuit between $s_k$ and $s_{k-1}$, and the depolarizing noise causes each path $f(s)$ to decay proportionally to $(1-p)^{|s|}$, where $|s|=|s_0|+|s_1|+...+|s_d|$ is the Pauli (Hamming) weight of the path $s$. Furthermore, since the unitary $U_k$ is decomposed in $D-$qubit gates, with $D=\Theta(1)$, it can be further decomposed in transition amplitudes of $D-$qubit gates, each of which can he computed in $O(1)$ time. Therefore, $f(s)$ can be computed in $O(|s|)=O(nd)$ time for any Pauli path $s$.

We can then rewrite everything summing together all the paths of the same weight:

\begin{equation}\label{eq:polynomial_O}
\left<O\right>_{\mathrm{err}}=\sum_w F_w (1-p)^w,
\end{equation}

\noindent 
where
\begin{align} \label{eq:F_w}
F_w=\sum_{|s|=w} f(s).
\end{align}

\noindent The classical algorithm to estimate $\left<O\right>_{\mathrm{err}}$, works by truncating all the Pauli paths of a weight greater than a cut-off $\ell$. 

\[
\left<O\right>_{\ell}=\sum_{w \leq \ell} F_w (1-p)^w.
\]

\noindent The error will then be given by the truncated terms, and is tunable by the cut-off paramenter $\ell$:

\[
\mathrm{Error}=\left|\left<O\right>_{\ell}-\left<O\right>_{\mathrm{err}}\right| = \left|\sum_{w > \ell} F_w (1-p)^w\right|.
\]

\section{T gates} \label{section:T_gates}

In this section we will consider circuits that are comprised only of $T$ gates and Clifford gates. We first give a sufficient condition for simulability of 2D circuits that is based exclusively on the distribution of $T$ gates in the cicuit. This is stated in theorem \ref{theorem:T_gates}. As a corollary of this result, in subsection \ref{subsection:random_model} we study a random model where the $T$ gates are inserted randomly throughout the circuit. We find that the classical simulation is possible depending on the rate at which $T$ gates are introduced. This is stated in theorem \ref{theorem:random_T}. The randomness in the model serves the purpose of guaranteeing that the $T$ gates are distributed evenly throughout the circuit.

To understand this result, first consider a quantum circuit that is comprised only of Clifford gates. Since Clifford gates map Pauli strings to Pauli strings, performing the Heisenberg evolution of an observable $O$ is trivial: the application of every layer will yield a single Pauli string. Hence, there is only one nonzero Pauli path, implying that the circuit can be classically simulated. We now consider a quantum circuit with Clifford+$T$ gates. In this case, there will be more than one nonzero Pauli path, since a single $T$ gate can map a Pauli string to a linear combination of two Pauli strings. Hence, the number of Pauli paths that need to be computed will depend on the number of $T$ gates that are applied. Recall from Eq. (\ref{eq:F_w}) that the contribution to the noiseless expectation value $\left<O\right>_{\mathrm{ideal}}$ of the Pauli paths of weight $w$ is given by
\[
(1-p)^w \sum_{|s|=w} f(s).
\]
\noindent In general, the sum $\sum_{|s|=w} f(s)$ can be exponentially large.
However, if the number of nonzero Pauli paths is exponentially large, but grows slower than $(1/(1-p))^w$, the term $(1-p)^w$ will dominate, and it will be possible to truncate the contributions of Pauli paths of large weight, allowing us to classically approximate $\left<O\right>_{\mathrm{err}}$. Our key observation is that this will be true in circuits where the $T$ gates are only a small fraction of the total number of gates, if these $T$ gates are sufficiently evenly distributed throughout the circuits.

We begin by formalizing what we mean by sufficiently evenly distributed $T$ gates. For concreteness, we will analyze 2D quantum circuits, which is also the circuit topology adopted in several experiments \cite{google2019quantumsupremacy,kim_ibm2023evidence}, but our analysis can be extended to higher dimensions. We consider a 2D quantum circuit $\mathcal{C}$ of depth $d$, comprised of (2-qubit and 1-qubit) Clifford and $T$ gates. The $n$ qubits are labeled by their coordinates $(x,y)$ in 2D, and are arranged in a square lattice. The circuit is specified by the coordinates $(x,y,t)$: that is $(x,y,t)$ refers to the gate that is applied to the qubit $(x,y)$ at the time step $t$. On this geometry, We consider subsets of points $\mathcal{A} \subseteq \{(x,y,t) \in \mathbb{N}: x \leq \sqrt{n},y \leq \sqrt{n}, t\leq D\}$. We will define by $T(\mathcal{A})$ the number of $T$ gates that are acting on qubits included in $\mathcal{A}$. We will also define by $G_\mathcal{A}$ the associated graph to $\mathcal{A}$ that is obtained by inserting edges between every point of $\mathcal{A}$ and its next nearest neighbors. This allows us to define the condition of $(Q,k)$-sparseness in the $T$ gates as follows:

\begin{definition}\label{definition:sparseness}
A circuit $\mathcal{C}$ on a square lattice is $(Q,k)$-sparse in the $T$ gates if,  $\forall \mathcal{A}$ such that $|\mathcal{A}|\geq k$ and $G_\mathcal{A}$ is connected,  $T(\mathcal{A})/|\mathcal{A}| \leq Q$.

\end{definition}

The condition of $(Q,k)$-sparseness ensures that, for all sufficiently large sets of contiguous gates, the fraction of $T$ gates is no larger than $Q$. This will allow us to upper-bound the number of $T$ gates that any nonzero Pauli path of weight $w \geq k$ encounters. Consequently, this allows us to upper-bound the number of nonzero Pauli paths of a given weight, denoted by $N_w$, as a function of $Q$: 

\begin{lemma}[Counting paths]\label{lemma:counting_geometric}
Consider a quantum circuit $\mathcal{C}$ on a square lattice which is  $(Q,k)$-sparse in the $T$ gates, and a Pauli observable $O$. Then, the total number of Pauli paths of weight $k \leq w \leq \ell $ with nonzero contribution can be upper-bounded by $\sum_{w=k}^{\ell} N_w \leq 2^{Q\ell}$. Furthermore, when $\ell>k$, the nonzero Pauli paths of weight $w \leq \ell$ can be enumerated in time $2^{Q\ell}+2^{k}$.
\end{lemma}
\begin{proof} 
We consider the set of different Pauli paths $s=(s_0,s_1,...,s_d)$, and would like to bound the number of nonzero paths with Pauli weight $k \leq w \leq \ell$. As defined in Eq.(\ref{eq:expression_Pauli_path}), $f(s)$ denotes the contribution from the Pauli path $s$. Here, $s_k \in \{I,X,Y,Z\}^{\otimes n}$ is the Pauli string at time $k$. The quantum circuit is defined on a square lattice, and the position of the qubits is labeled by the coordinates $(x,y)$. We will therefore denote by $s_k(x,y) \in \{I,X,Y,Z\}$ the operator of the Pauli string $s_k$ acting on the position $(x,y)$. Likewise, the quantum circuit is specified by the coordinates $(x,y,t)$: that is, a point $(x,y,t)$ refers to the gate that is applied to the qubit $(x,y)$ at the time step $t$.

Given the Pauli string $s_k$ at time $k$, we define a set $\mathcal{S}_{k-1}(s_k)$ containing all the different Pauli strings $s_{k-1}$ that have a nonzero contribution to the transition amplitude: 
\[
\mathcal{S}_{k-1}(s_k)=\bigg\{s_{k-1} \in \{I,X,Y,Z\}^{\otimes n} | \mathrm{tr}\left(U_k^{\dagger}s_kU_k s_{k-1}\right) \neq 0\bigg\}.
\]

\noindent By definition, from Eq. (\ref{eq:expression_Pauli_path}), the contribution $f(s)$ of a Pauli path with contains a string $s_k$ which is outside the set $\mathcal{S}_{k-1}(s_{k})$ is $f(s)=0$. Furthermore, since the observable $O$ itself is a Pauli observable, according to Eq. (\ref{eq:expression_Pauli_path}) it must be $s_d=O$, because otherwise one obtains $f(s)=0$.  In fact, this provides a necessary condition for a Pauli path to have nonzero contribution:

\[
f(s) \neq 0 \iff s_d=O, s_{d-1} \in \mathcal{S}_{d-1}(s_d), s_{d-2} \in \mathcal{S}_{d-2}(s_{d-1}),\dots, s_0 \in \mathcal{S}_0(s_1).
\]

\noindent Then, the number of nonzero paths of weight $w$ such that $k \leq w \leq \ell$ is given by the sum

\begin{align}\label{eq:Nw_sum}
\sum_{w=k}^{\ell}N_w= \sum_{\substack{s_d=O,s_{d-1} \in \mathcal{S}_{d-1}(s_d), ... ,s_0 \in \mathcal{S}_0(s_1) \\
                               k < |s_d|+...+|s_0| \leq \ell}}
          1.
\end{align}

\noindent One can notice that the Clifford gates cannot increase the number of nonzero Pauli paths since, by definition, the Clifford gates map Pauli strings to Pauli strings. That is, if the layer $U_k$ contains only Clifford gates, the set $\mathcal{S}_{k-1}(s_k)$ only contains one element, which is the Pauli string $U_k^{\dagger}s_kU_k$. On the other hand, a $T$ gate can branch a Pauli operator into two. Specifically, for a given $s_k$, the size of the set $\mathcal{S}_{k-1}(s_{k})$ will be bounded by

\begin{equation}\label{eq:size_sk}
\left|\mathcal{S}_{k-1}(s_{k})\right|=\sum_{s_{k-1} \in \mathcal{S}_{k-1}(s_k)}1 \leq 2^{T_k(s_k)},
\end{equation}
\noindent where $T_k(s_k)$ counts the number of $T$ gates in the layer $k$ of the quantum circuit that are applied on a position $(x,y)$ such that $s_k(x,y) \neq I$. That is $T_k(s_k)$ counts the number of $T$ gates that the Pauli path $s$ ``encounters'' at time $k$, where we define encountering as applying the $T$ gate to a Pauli operator, but not to the identity.

Then, from the sparseness condition one can deduce that, for any Pauli path of weight $w$ with $k \leq w \leq \ell$ and nonzero contribution:

\begin{align}\label{eq:T_gates_Q}
T_1(s_1)+ ... + T_d(s_d) \leq Q \left( |s_1| + ... + |s_d| \right) \leq Qw \leq Q \ell.
\end{align}

\noindent To see this, for a given Pauli path $s$ we will consider the set of points $\mathcal{A}(s)$ that contains the Pauli operators but not the identities.
\[
\mathcal{A}(s)=\{(x,y,t) | s_t(x,y) \neq I\}.
\]

\noindent We can then define the associated graph $G_{\mathcal{A}(s)}$, and realize that it must be connected. This is a consequence of the quantum gates of the circuit being local, and of the fact that a quantum gate cannot map a Pauli operator to the identity.  For example, We consider the case where $s_d(x_1,y_1)=X$ for some qubit $(x_1,y_1)$. If the layer $U_d$ of the quantum circuit applies a $1$-qubit gate at qubit $(x_1,y_1)$, then $s_{d-1}(x_1,y_1) \neq I \forall s_{d-1} \in \mathcal{S}_{d-1}(s_d)$. On the other hand, if a $2$-qubit gate is applied, since it must be local, there exists a qubit $(x_2,y_2)$ that is a nearest neighbor of $(x_1,y_1)$ such that $s_{d-1}(x_2,y_2) \neq I$. That is, $|x_2-x_1|+|y_2-y_1| \leq 1$. This is because a $2$-qubit gate can never map the operator $X \otimes R$ to $I \otimes I$ (where $R \in \{I,X,Y,Z\}$ ). Therefore, if point $(x_1,y_1,d) \in \mathcal{A}(s)$, and $s$ is a Pauli path with nonzero contribution, there is a point $(x_2,y_2,d-1) \in \mathcal{A}(s)$ that is a next nearest neighbor of $(x_1,y_1,d)$. Repeating this argument shows that $G_{\mathcal{A}(s)}$ must be connected in order to have a nonzero Pauli path. Then, because of the definition of $(Q,k)$-sparseness, one obtains that $T(\mathcal{A})\leq Q|\mathcal{A}|$, which is equivalent to Eq.(\ref{eq:T_gates_Q}).

That is, due to the sparseness of $T$ gates in the circuit, the fraction of $T$ gates that any Pauli path of weight $w \geq k$ will encounter will be no greater than $Q$, with the rest being Clifford gates. This allows us to bound the sum in Eq. (\ref{eq:Nw_sum}):

\begin{align*}
\sum_{w=k}^{\ell}N_w &= \sum_{\substack{s_d=O,s_{d-1} \in \mathcal{S}_{d-1}(s_d), ... ,s_0 \in \mathcal{S}_0(s_1) \\
                               k \leq |s_d|+...+|s_0| \leq \ell }}
          1 \\
&\overset{(1)}{\leq} \sum_{\substack{s_d=O,s_{d-1} \in \mathcal{S}_{d-1}(s_d), ... ,s_0 \in \mathcal{S}_0(s_1) \\
                               k \leq |s_d|+...+|s_0| \leq \ell}}
          2^{Q\ell -T_1(s_1)-T_2(s_2)- ... -T_d(O)} \\
  & \overset{(2)}{\leq} \sum_{\mathcal{S}_{d-1}(O)} ... \sum_{\mathcal{S}_{k}(s_{k+1})} ...  \sum_{\mathcal{S}_{1}(s_2)} \sum_{\mathcal{S}_{0}(s_1)} 2^{Q \ell-T_1(s_1)-T_2(s_2)...-T_d(O)}\\
& \overset{(3)}{\leq} \sum_{\mathcal{S}_{d-1}(O)} ... \sum_{\mathcal{S}_{k}(s_{k+1})} ...  \sum_{\mathcal{S}_{1}(s_2)}2^{Q\ell-T_2(s_2)-...-T_d(O)} \\
& \vdots \\
& \overset{(3)}{\leq} \sum_{\mathcal{S}_{d-1}(O)}2^{Q\ell-T_d(O)} \leq 2^{Q\ell},
\end{align*}

\noindent where in $(1)$ we have used Eq.(\ref{eq:T_gates_Q}), in $(2)$ we have dropped the restriction $k \leq |s_d|+...+|s_0| \leq \ell$ in the sum, and in $(3)$ we have (repeateadly) used Eq. (\ref{eq:size_sk}).

Furthermore, the nonzero paths can be enumerated in a straightforward way: for a given $s_{k}$, the elements of $\mathcal{S}_{k-1}$ can be obtained by just applying the relevant $T$ gates (which yield at most $2^{T_k(s_{k})}$ configurations), and Clifford gates (which only yield one configuration). The paths with weight $w \leq k $ can therefore be enumerated in time $2^{k}$, while the paths with weight $k \leq w \leq \ell$ can be enumerated in time $2^{Q \ell}$ following the argument above. This yields $2^{k}+ 2^{Q\ell}$ time in total. \end{proof}

As a consequence of the refined counting of the paths, one can show that the algorithm succeeds for a sufficiently high error rate:

\begin{theorem}\label{theorem:T_gates}
 Consider a 2D quantum circuit comprised of Clifford and $T$ gates that is $(Q,a \log n)$-sparse in the $T$ gates for some constant $a=\Theta(1)$, an initial product state $\rho_0$, and an observable $O=\sum_{k=1}^{g(n)} a_k O_k$ given by a linear combination of at most $g(n)=O(\mathrm{poly}(n))$ Pauli observables, where the $a_k$ and $O_k$ are known. Then, there exists a classical algorithm that can estimate the noisy expectation value $\left<O\right>_{\mathrm{err}}$ up to precision $\epsilon \|O\|$ in time $T=O\left(\mathrm{poly}\left(n,1/\epsilon \right)\right)$, $\forall p>1-1/2^Q=\Theta(Q)$ as $Q \to 0$. Specifically, for sufficiently small $p,Q$, the time complexity is 
\[
T=g(n) \times O\left[\max \left(  \left(\left(ng(n)d/\epsilon\right)^{\frac{Q}{p -Q \log 2}} + n^{a \log 2}\right)\times \frac{\log (ng(n)d/\epsilon)}{p - Q\log 2},n^{a \log 2} \log n\right)\right].
\]

\end{theorem}
\begin{proof} 
We can first assume that the observable $O$ is a single Pauli string, since the the result for a linear combination of Pauli string follows can then be obtained by simply applying the algorithm to each of the $O_k$.
From Lemma \ref{lemma:counting_geometric}, one obtains that the number of Pauli paths of weight $w$ is bounded by $N_w \leq 2^{Qw}$, when $w \geq a \log n$. Therefore, since $\left|f(s)\right| \leq 1 \forall s$ due to unitarity, one gets $\left|F_w\right| \leq N_w$, and:

\[
\left| \sum_{w>\ell} F_w (1-p)^w \right| \leq \sum_{w>\ell} \left| N_w\right| (1-p)^w  \leq \sum_{w=\ell}^{nd} 2^{wQ}(1-p)^w \leq nd\left[2^Q(1-p)\right]^\ell \leq \epsilon.
\]

\noindent To fulfill this condition it suffices to pick a cut-off

\[
\ell=\max\left(\frac{\log (nd/ \epsilon)}{\log \left[1/\left(2^Q(1-p)\right)\right]},a \log n\right).
\]

\noindent Since any $f(s)$ can be computed in $O(|s|)$ time, the runtime would be 

\[
T=O\left[\ell \times \left(2^{Q \ell}+n^{a \log 2}\right)\right]=O\left(\mathrm{poly}\left(n,1/\epsilon \right)\right).
\]

\noindent Specifically, if $p \ll 1$ the runtime is
\[
T=O\left[\max \left(  \left(\left(nd/\epsilon\right)^{\frac{Q}{p -Q \log 2}} + n^{a \log 2}\right)\times \frac{\log (nd/\epsilon)}{p - Q\log 2},n^{a \log 2} \log n\right)\right].
\]
We now consider that $O$ is given by a linear combination of $g(n)=O(\mathrm{poly(n)})$ Pauli observables $O_k$, $O=\sum_{k=1}^{g(n)}a_k O_k$. Using the method above, one can approximate each of the $O_k$ individually up to precision $\epsilon^{\prime}$. That is, the classical algorithm computes $\left<O_k\right>_{\ell}$ such that $\left|\left<O_k\right>_{\ell}-\left<O_k\right>_{\mathrm{err}}\right| \leq \epsilon^{\prime}$. The classical algorithm then outputs $\left<O\right>_{\ell}=\sum_k a_k \left<O_k\right>_{\ell}$. The operator norm $\|O\|$ can be lower bounded using

\[
\|O\| \geq \|O\|_F/\sqrt{2^n} =\sqrt{\sum_k a_k^2}.
\]

\noindent And the total error can be upper bounded by
\[
\left| \left<O\right>_{\mathrm{err}}-\left<O\right>_{\ell}\right| \leq \sum_{k=1}^{g(n)}|a_k| \epsilon^{\prime} \leq \epsilon^{\prime} \sqrt{g(n)}\sqrt{\sum_k a_k^2} \leq \epsilon^{\prime} \sqrt{g(n)}\|O\|,
\]

\noindent One can then pick $\epsilon^{\prime}=\epsilon/\sqrt{g(n)}$, yielding a runtime:

\[
T=g(n) \times O\left[\max \left(  \left(\left(ng(n)d/\epsilon\right)^{\frac{Q}{p -Q \log 2}} + n^{a \log 2}\right)\times \frac{\log (ng(n)d/\epsilon)}{p - Q\log 2},n^{a \log 2} \log n\right)\right].
\]

\end{proof}
For reference, we can compare this result to a brute-force simulation in the Pauli basis which computes all the Pauli paths. For a 2D circuit that is also $(Q,k)-$sparse, from Lemma \ref{lemma:counting_geometric} it follows that the total number of Pauli paths of Pauli weight less or equal than $\ell$ can be enumerated in time $2^{Q \ell}+2^k$. Furthermore, due to geometric locality constraints, for an initial Pauli observable of $O(1)$ Pauli weight, the maximum weight of a Pauli path is $\ell=O(d^{3})$, where $d$ is the circuit depth. Hence, even in the noiseless case, the expectation value of $O$ can be computed in $2^{QO(d^{3})}$ time. Hence, for small $Q$, the classical simulation becomes easier due to the smaller number of $T$ gates. However, for constant $Q$, this runtime is superpolynomial for sufficiently large circuit depth,  $d=\omega\left(\sqrt{\log n}\right)$. Consequently, Theorem \ref{theorem:T_gates} shows that the presence of noise allows us to achieve a superpolynomially better performance with the Pauli-path method compared to the noiseless case where we would essentially have to keep track of all the Pauli paths.

\subsection{Random model}\label{subsection:random_model}

Here, we study a random model that introduces $T$ gates randomly throughout the circuit. Since the $T$ gates are introduced randomly, they will be evenly distributed throughout the circuit, which relaxes the geometric conditions for simulability that were detailed above. As a consequence, for this random model, the threshold error rate above which the simulation becomes easy will depend only on the rate at which $T$ gates are introduced in the circuit.

\begin{definition}[Random model] \label{random_model}
We consider circuits of a fixed architecture, and a gate set that contains only Clifford (1-qubit and 2-qubit) gates and $T$ gates. The circuits are generated according to the following rule: In every position of the architecture, one can choose which gate to place. In the $2$-qubit gate positions of the architecture, any Clifford gate can be placed. In the $1$-qubit gate positions, either a Clifford or a $T$ gate can be placed. However, this is done randomly: with probability $Q$ one can pick any gate from the gate-set, and with probability $1-Q$ one can only pick a Clifford gate. 
\end{definition}

In the random model defined above the parameter $Q$ controls the number of $T$ gates that can be introduced on average. We remark that this model allows significant freedom to create circuits, since it only imposes an upper bound on the probability that a given gate is a $T$ gate, but otherwise it leaves freedom to choose which gates to apply to each qubit. Therefore, its behavior will deviate significantly from other random distributions (e.g. Haar random circuits) that are typically studied \cite{Aharonov2023_paulipaths,haah2018_operatorspreading,bouland2019complexity}. For example, the circuits sampled from this distribution do not satisfy gate-set orthogonality. Recall that a random distribution $\mathcal{D}$ is said to satisfy gate-set orthogonality if, for circuits sampled from $\mathcal{D}$, the average over the product of two different Pauli paths is zero \cite{Aharonov2023_paulipaths} i.e.
\[
\underset{\mathcal{C} \sim \mathcal{D}_Q}{\mathds{E}}\left[f(s)f(s^{\prime})\right]=0 \text{ for } s\neq s^{\prime}.
\]
To see that our model (Definition \ref{random_model}) does not satisfy gate-set orthogonality, it is enough to consider the following setting: we consider a single qubit in the state $(\ket{+Y}+\ket{+})/\sqrt{3}$, with $\ket{+Y}=(\ket{0}+i \ket{1})/\sqrt{2}$ (i.e.~the $+1$ eigenstate of $Y$). With probability $Q$ a $T$ gate is applied, while with probability $1-Q$ the identity is applied. Finally, the observable $X$ is measured. If the identity is applied, then the only nonzero Pauli path is $s=(X,X)$, which we denote as $f((X,X),I)=2/3$. On the other hand, since $T^{\dagger}XT=(X-Y)/\sqrt{2}$, if the $T$ gate is applied there are two nonzero Pauli paths: $s=(X,X)$  and $s=(Y,X)$, denoted as $f((X,X),T)=\sqrt{2}/3$ and $f((Y,X),T)=-\sqrt{2}/3$. Therefore, if $s=(X,X)$ and $s^{\prime}=(Y,X)$ one gets
\begin{equation}\label{eq:orthogonality}
\mathds{E}_{\mathcal{C}} f(s)f(s^{\prime})=(1-Q)f(s,I)f(s^{\prime},I) +Qf(s,T)f(s^{\prime},T)=-2Q/9,
\end{equation}
implying that the gate-set orthogonality is violated.

Despite not satisfying gate-set orthogonality, one can still upper-bound the number of nonzero paths of a given Pauli weight $w$ as a function of $Q$:

\begin{lemma} [Counting paths]\label{lemma:counting_averaged}

Consider a quantum circuit $\mathcal{C}$ drawn from a distribution $\mathcal{D}_Q$ constructed according to Definition \ref{random_model} with parameter $Q$. Then, the total number of nonzero Pauli paths of weight $w \leq \ell $  with nonzero contribution, on average, can be upper-bounded by
\begin{align}
\underset{\mathcal{C} \sim \mathcal{D}_Q}{\mathds{E}} \sum_{w\leq \ell}N_w \leq \left(1+Q\right)^w.
\end{align}
\end{lemma}
\noindent Furthermore, the paths of weight $w \leq \ell$ can be enumerated in time $\left(1+Q\right)^\ell$, on average.
\begin{proof} 

This can be shown adapting the proof of Lemma \ref{lemma:counting_geometric}. One wants to average over the sum

\begin{align}
\mathds{E}_{\mathcal{C}} \sum_{w \leq \ell} N_{w}= \mathds{E}_{\mathcal{C}} \sum_{\substack{s_d=O,s_{d-1} \in \mathcal{S}_{d-1}(s_d), ... ,s_0 \in \mathcal{S}_0(s_1) \\
                               |s_d|+...+|s_0| \leq \ell}}
          1. 
\end{align}

\noindent In this case, we notice that, averaging over Eq. (\ref{eq:size_sk})

\[
\mathds{E}_{\mathcal{C}}\left|\mathcal{S}_{k-1}(s_{k})\right| \leq \mathds{E}_{\mathcal{C}} 2^{T_k(s_k)}.
\]

\noindent To perform the average, we realize that, for circuits sampled from the distribution $\mathcal{D}_Q$, by definition, the probability that a gate at a given position is a $T$ gate is upper bounded by $Q$. Therefore:

\begin{align}\label{eq:relation_sk_average}
\mathds{E}_{\mathcal{C}} 2^{T_k(s_k)} \leq \sum_{j=0}^{|s_{k}|} Q^j(1-Q)^{|s_{k}|-j} {|s_{k}| \choose j} = (1+Q)^{|s_{k}|}.
\end{align}

\noindent And hence

\begin{align} \label{eq:size_random}
\mathds{E}_{\mathcal{C}} \sum_{s_{k-1} \in \mathcal{S}_{k-1}(s_{k})}1 \leq (1+Q)^{|s_k|}.
\end{align}

\noindent Following the rest of the argument in Lemma \ref{lemma:counting_geometric} with Eq. \ref{eq:relation_sk_average} proves the Lemma:

\begin{align*}
\mathds{E}_{\mathcal{C}}  \sum_{w \leq \ell} N_{w} &= \mathds{E}_{\mathcal{C}} 
 \sum_{\substack{s_d=O,s_{d-1} \in \mathcal{S}_{d-1}(s_d), ... ,s_0 \in \mathcal{S}_0(s_1) \\
                               |s_d|+...+|s_0|\leq \ell}}
          1 \\
 &\overset{(1)}{\leq} \mathds{E}_{\mathcal{C}} 
 \sum_{\substack{s_d=O,s_{d-1} \in \mathcal{S}_{d-1}(s_d), ... ,s_0 \in \mathcal{S}_0(s_1) \\
                               |s_d|+...+|s_0| \leq \ell}}
          (1+Q)^{\ell-|s_1|-|s_2|- ... -|s_d|} \\
 & \overset{(2)}{\leq} \mathds{E}_{\mathcal{C}} 
 \sum_{\mathcal{S}_{d-1}(O)} ... \sum_{\mathcal{S}_{k}(s_{k+1})} ...  \sum_{\mathcal{S}_{1}(s_2)} \sum_{\mathcal{S}_{0}(s_1)} (1+Q)^{\ell-|s_1|-|s_2|...-|O|)}\\
 & \overset{(3)}{\leq} \mathds{E}_{\mathcal{C}} 
 \sum_{\mathcal{S}_{d-1}(O)} ... \sum_{\mathcal{S}_{k}(s_{k+1})} ...  \sum_{\mathcal{S}_{1}(s_2)}(1+Q)^{\ell-|s_2|-...-|O|} \\
 & \vdots \\
&\overset{(3)}{\leq} \mathds{E}_{\mathcal{C}}  \sum_{\mathcal{S}_{d-1}(O)}(1+Q)^{\ell-|O|} \leq (1+Q)^{\ell},
\end{align*}
\noindent where in $(1)$ we have used that $(|s_1|+...+|s_d| \leq \ell)$, in $(2)$ we have dropped the restriction $|s_0|+...+|s_d| \leq \ell$ in the sum, and in $(3)$ we have used Eq. (\ref{eq:size_random}) repeatedly.

The enumeration algorithm is also the same as in Lemma \ref{lemma:counting_geometric}, taking on average $(1+Q)^\ell$ time.
\end{proof}

As a consequence of this, if the rate at which $T$ gates are introduced is sufficiently small, the classical algorithm succeeds with high probability:

\begin{mythm}{4.1}\label{theorem:random_T}

Consider a quantum circuit $\mathcal{C}$ drawn from a distribution $\mathcal{D}_Q$ constructed according to Definition \ref{random_model} with parameter $Q$, an initial product state $\rho_0$, and an observable $O=\sum_{k=1}^{g(n)} a_k O_k$ given by a linear combination of at most $g(n)=O(\mathrm{poly}(n))$ Pauli observables, where the $a_k$ and $O_k$ are known. Then, $\forall p>1-1/(1+Q)=\Theta(Q)$ as $Q \to 0$, there exists a classical algorithm that can estimate the noisy expectation value $\left<O\right>_{\mathrm{err}}=\mathrm{tr}\left(O\Phi_{\mathcal{C},p}(\rho_0)\right)$ up to precision $\epsilon \|O\|$ with a success probability of at least $1-\delta$ in time $T=O\left(\mathrm{poly}\left(n,1/\epsilon,1/\delta \right)\right)$. Moreover, for sufficiently small $p,Q$, it is
\[
T=g(n) \times O\left[ \left(\frac{ng(n)d}{\epsilon \delta}\right)^{\frac{Q}{p-Q}} \left(\frac{\log \frac{ng(n)d}{\delta \epsilon}}{p-Q}\right)\right].
\]

\end{mythm}

\begin{proof}
We can first assume that the observable $O$ is a single Pauli string, since the result for a linear combination of Pauli string follows can then be obtained by simply applying the algorithm to each of the $O_k$.
From the counting argument in Lemma \ref{lemma:counting_averaged}, the number of Pauli trajectories of a given weight $w$ can be upper bounded by $\mathds{E}_\mathcal{C} N_w \leq \left(1+Q\right)^w$ on average. The classical algorithm consists on computing the contribution of all the paths up to a weight $\ell=O(\log n)$, denoted by  $\left<O\right>_{\ell}$. One can write 
\[
\left| \left<O\right>_{\mathrm{err}}-\left<O\right>_{\ell}\right|=\sum_{w>\ell}F_w \left(1-p\right)^w.
\]

\noindent From unitarity, since $\|O\|=1$, it must be $|f(s)| \leq 1 \forall s$. One can then upper bound 
\[
\left|F_w\right| = \left| \sum_{|s|=w}f(s)\right| \leq N_w.
\]

\noindent And therefore, since $\mathds{E}_\mathcal{C} N_w\leq \left(1+Q\right)^w$, one can then bound

\[
\mathds{E}_\mathcal{C}\left| \sum_{w=\ell}^{nd} (1-p)^w F_w\right| \leq  \sum_{w=\ell}^{nd} (1-p)^w \mathds{E}_\mathcal{C} \left| F_w (\mathcal{C})\right| \leq nd (1-p)^\ell (1+Q)^\ell.
\]

\noindent Therefore, denoting $c=(1-p)(1+Q)$ the error for a given cut-off $\ell$ is upper bounded by $ndc^\ell$, which will be small when $c<1$.

Then, by Markov's inequality

\[
\mathrm{Pr}(\left|\left<O\right>_{\mathrm{err}}-\left<O\right>\right|_{\ell} > \epsilon) \leq \frac{ndc^\ell}{\epsilon}=\delta.
\]
To fulfill this condition it suffices to pick a cut-off

\[
\ell=\frac{\log (nd/\delta \epsilon)}{\log 1/c}.
\]

\noindent The runtime would be 

\[
T=O\left[\ell \times (1+Q)^\ell\right]=O\left(\mathrm{poly}\left(n,1/\epsilon,1/\delta \right)\right).
\]

\noindent Specifically, for $p,Q \ll 1$ it is

\[
T=O\left[ \left(\frac{nd}{\epsilon \delta}\right)^{\frac{Q}{p-Q}} \left(\frac{\log \frac{nd}{\delta \epsilon}}{p-Q}\right)\right].
\]

Finally, We now consider that $O$ is given by a linear combination of $g(n)=O(\mathrm{poly(n)})$ Pauli observables $O_k$, $O=\sum_{k=1}^{g(n)}O_k$. Using the method above, one can approximate each of them individually up to precision $\epsilon^{\prime}$. Using the same argument as in the proof of theorem \ref{theorem:T_gates}, and denoting the output of the classical algorithm by $\left<O\right>_{\ell}$, the error is upper bounded by
\[
\left| \left<O\right>_{\mathrm{err}}-\left<O\right>_{\ell}\right| \leq \sum_{k=1}^{g(n)}|a_k| \epsilon^{\prime} \leq \epsilon^{\prime} \sqrt{g(n)}\sqrt{\sum_k a_k^2} \leq \epsilon \sqrt{g(n)}\|O\|,
\]

\noindent One can then pick $\epsilon^{\prime}=\epsilon/\sqrt{g(n)}$, yielding a runtime:

\[
T=g(n) \times O\left[ \left(\frac{ng(n)d}{\epsilon \delta}\right)^{\frac{Q}{p-Q}} \left(\frac{\log \frac{ng(n)d}{\delta \epsilon}}{p-Q}\right)\right].
\]

\end{proof}

Note that, in the small $Q$ limit, the threshold condition reduces to $p>Q$, highlighting the fact that the classical simulation becomes easy when noise is introduced at a faster rate than $T$ gates.

\section{Variational circuits}\label{section:Ising_hamiltonian}

Here we apply the results of section \ref{section:T_gates} to variational circuits that attempt to find the ground energy of classical Ising Hamiltonians (for example, for solving combinatorial optimization problems). We show that there can be no asymptotic exponential quantum speedup in the system size with these circuits. In Lemma \ref{lemma:Ising_hamiltonian_hardness} it is shown that for the hard instances of the problem, there is a classical efficient algorithm that approximates the output energy of the quantum computation. On the other hand, Lemma \ref{lemma:easy_instances} shows that the ground energy of easy instances (i.e the ones with planar or close to planar graphs) can be efficiently approximated with a classical algorithm.
 
 We consider a classical Ising Hamiltonian on some graph $G=(V,E)$ with bounded degree $\Delta = \Theta(1)$.

 \begin{align}\label{eq:ising_hamiltonian}
 H=\sum_{(i,j) \in E}J_{ij}Z_i Z_j + \sum_i b_i Z_i,
 \end{align}

\noindent The only assumption that we make on the couplings is that they are non-zero $\Theta(1)$ constants i.e. $\max_{i,j}|J_{ij}|, \min_{i,j}|J_{ij}| = \Theta(1)$ and $J_{i, j}\neq 0 \ \forall \ (i, j) \in E$, and $\max_i |b_i|=\Theta(1)$.

We will now define the type of circuits that we will study. Variational circuits usually consist on the application of $m$ variational layers, $L_1,...,L_m$, where each layer is given by $L_j=e^{-i\gamma_j H}e^{-i\alpha_j H_M}$, where $H$ is the target Hamiltonian, $H_M$ is a mixer Hamiltonian, and $\gamma_{j},\alpha_{j}$ are optimization parameters. The mixer Hamiltonian is often constructed in QAOA as $H_M=\sum_i X_i$, but can also be more general. We refer to the extensive literature on variational algorithms for more details \cite{cerezo2021variational,QAOA_review_2024,farhi2014qaoa}. Our only assumption for the mixer Hamiltonian is that the term $e^{-i\alpha_j H_M}$ can be implemented by a constant depth circuit. In our construction, which is sketched in Fig. \ref{fig:variational_ansatz} we will assume that the cost Hamiltonian is implemented as follows:

\begin{align}\label{eq:variational_layer}
e^{-i\gamma_{j} H}=V_{j,1} S_1 V_{j,2} S_2 ... V_{j,M-1} S_{M-1} {V_{j,M}}.
\end{align}
Here, we denote by $V_{j,k}$ the computing unitaries, and by $S_{k}$ the permuting unitaries. With this notation, the $V_{j,k}$ are comprised of (possibly several layers of) 2-qubit local unitaries and implement the terms $e^{-i \gamma Z_{i_1} Z_{i_2}}$ and $e^{-i \gamma Z_{i}}$ of the Hamiltonian, while the $S_k$ are comprised of layers of SWAP gates. The $S_{k}$ are only necessary if the graph $G(n,E)$ does not match the geometry of the quantum computer. Since $G$ has bounded degree $\Delta=\Theta(1)$, a constant depth of computing unitaries is enough to implement the circuit for any bounded degree Ising Hamiltonian of the form of Eq. (\ref{eq:ising_hamiltonian}), $\sum_k \mathrm{depth}(V_{j,k})=O(1) \forall j$. Besides this, we do not place any restriction on how the circuit is constructed, or how many layers are used to construct the $V_{j,k}$, $S_{k}$, or $L_j$. However, in order to implement any Ising Hamiltonian the permuting unitaries $S_{k}$ might need a large depth in general. This motivates the definition of the embedding overhead factor $\lambda$, which will measure the depth overhead that is required to embed the graph onto a 2D architecture:

\begin{definition}\label{definition:variational_circuit}
Given a variational 2D quantum circuit constructed according to the prescription in Eq. (\ref{eq:variational_layer}), the embedding overhead factor of the circuit is defined as
\[
\lambda=\sum_k \mathrm{depth}(S_{k}).
\]
\end{definition}
\noindent We note that if the graph of the Hamiltonian matches the geometry of the quantum computer it is $\lambda=0$, since no embedding is required in that case. In the most general case it is $\lambda=O(n)$ \cite{harrigan2021QAOA_nonplanar}. Finally, we note that we consider a gate-set whose only non Clifford gates are $T$ gates, but the analysis can easily be extended to other gate-sets.

The problem can be formulated in terms of the embedding overhead factor $\lambda$. If the embedding overhead for a given circuit is $O(1)$, then the graph $G$ on which the Hamiltonian is defined is almost planar, and following the techniques in \cite{bansal2008classical_Ising} there exists a polynomial time classical algorithm that approximates the ground energy of the Hamiltonian, although its runtime is exponential with the precision. On the other hand, if $\lambda=\omega(1)$ for the quantum circuit, the results of section 3 imply that there is an efficient classical algorithm that can approximate the output of the quantum computer, since the number of Clifford gates will dominate the number of non-Clifford gates. This trade-off is formalized in the following Lemmas:

\begin{lemma}\label{lemma:easy_instances}
Consider an Ising Hamiltonian defined on a graph $G=(n,E)$ with bounded degree $\Delta = \Theta(1)$,

 \[
 H=\sum_{(i,j) \in E}J_{ij}Z_i Z_j + \sum_i b_iZ_i,
 \]
\noindent and a 2D quantum circuit $\mathcal{C}$ which attempts to find the ground energy of $H$, constructed according to Eq. (\ref{eq:variational_layer}), with an embedding overhead factor defined as in Definition \ref{definition:variational_circuit}. Then, if the overhead factor of the circuit is $\lambda=O(1)$ there is a classical algorithm that approximates the ground energy of $H$ up to absolute error $\epsilon n$ in time $T=O(n 2^{1/\epsilon^2})$.
\end{lemma}

\begin{proof} This can be shown by adapting the algorithm in \cite{bansal2008classical_Ising}. We will first show the result in the easiest case with $\lambda=0$, which was already obtained in \cite{bansal2008classical_Ising}, and then generalize it to $\lambda=O(1)$.

We assume that $\lambda=0$. In this case $H$ is already defined on a 2D lattice and no embedding is necessary. One can then define the Hamiltonian

\[
H^{\prime}=\sum_{(i,j) \in E^{\prime}}J_{ij}Z_i Z_j + \sum_i b_iZ_i,
\]

\noindent where the set of edges $E^{\prime}$ is obtained by subdividing the 2D lattice into $n/L^2$ blocks of size $L \times L$, and omitting the edges of $E$ that connect the sub-blocks. The ground state of $H^{\prime}$ is denoted by 
$\ket{\psi^{\prime}}$, and its ground energy by $E_0^{\prime}=\bra{\psi^{\prime}}H^{\prime}\ket{\psi^{\prime}}$, while the ground state of $H$ is denoted by 
$\ket{\psi}$, and its ground energy by $E_0=\bra{\psi}H\ket{\psi}$.

Then, since the Hamiltonian $H^{\prime}$ is comprised of disjoint sub-blocks of size $L \times L$, it follows that $E_0^{\prime}$ can be computed in time $T=O(2^{L^2} \times n / L^2)$ by solving each sub-block independently. Furthermore, since $H^{\prime}$ is obtained by omitting at most $4L$ edges of $H$ in every sub-block, and there are $n/L^2$ sub-blocks, one obtains that for every state $\ket{\phi}$:

\[
\left|\bra{\phi}H\ket{\phi}-\bra{\phi}H^{\prime}\ket{\phi}\right| \leq 4L J_\text{max} \frac{n}{L^2}.
\]

\noindent Then, if $E_0^{\prime}>E_0$ one finds

\[
E_0^{\prime} - E_0 \leq \bra{\psi}H^{\prime}\ket{\psi}-E_0 = \bra{\psi}H^{\prime}\ket{\psi}-\bra{\psi}H\ket{\psi} \leq 4J_{\max}\frac{n}{L},
\]
\noindent while, if if $E_0^{\prime}<E_0$
\[
E_0 - E_0^{\prime} \leq \bra{\psi^{\prime}}H\ket{\psi^{\prime}} -E_0^{\prime} = \bra{\psi^{\prime}}H\ket{\psi^{\prime}}-\bra{\psi^{\prime}}H^{\prime}\ket{\psi^{\prime}}\leq 4J_{\max}\frac{n}{L}.
\]

\noindent Hence, it is enough to pick

\[
L=\frac{4J_{\mathrm{max}}}{\epsilon}
\]
\noindent to approximate the ground energy $E_0$ up to an absolute error $\epsilon W$, with runtime $T=O(n 2^{1/\epsilon^2})$.

We now consider the case where $\lambda=O(1)$. Again, we subdivide the 2D lattice in sub-blocks of size $L \times L$, and the Hamiltonian $H^{\prime}$ is defined in the same way. In this case, the graph $G$ can contain edges that are not nearest neighbor in this 2D geometry. However, the condition $\lambda=O(1)$ places a restriction on how long the edges can be. Specifically, there exists a constant $c=O(\lambda)=O(1)$ such that no edge can connect points that are further than $c$ in the lattice. That is, points that are separated by a distance $\omega(1)$ cannot be connected by an edge. That is because, if they were, it would be impossible to embed the graph with only a constant number of SWAP layers. Therefore, for every sub-block of size $L \times L$ we are truncating edges of at most size $c$. This implies that we are truncating edges that connect points that are at most at a distance $c$ from the boundary of the sub-block. Since there are $8Lc$ such points, this means we are truncating at most $8Lc \Delta$ edges. This implies that, for every state $\ket{\phi}$:

\[
\left|\bra{\phi}H\ket{\phi}-\bra{\phi}H^{\prime}\ket{\phi}\right| \leq8Lc \Delta J_\text{max} \frac{n}{L^2}.
\]

\noindent Following the same argument as above yields

\[
\left|E_0^{\prime} - E_0\right| \leq 8cJ_{\max}\frac{n}{L}.
\]
It is then enough to pick $L=8cJ_{\max}/\epsilon=O(\lambda/\epsilon)$ to ensure that the relative error is smaller than $\epsilon$. Finally, since a sub-block contains $L^2$ points, its ground energy can be found in time $O(2^{L^2})=O(2^{1/\epsilon^2})$. The total runtime is then $T=O(n 2^{1/\epsilon^2})$.

\end{proof}

On the other hand, if a variational circuit is such that its embedding overhead is large, the Clifford gates will dominate:

\begin{lemma}\label{lemma:Ising_hamiltonian_hardness}
Consider an Ising Hamiltonian defined on a graph $G=(n,E)$ with bounded degree $\Delta = \Theta(1)$,

 \[
 H=\sum_{(i,j) \in E}J_{ij}Z_i Z_j + \sum_i b_iZ_i,
 \]
\noindent and a 2D quantum circuit $\mathcal{C}$ which attempts to find the ground energy of $H$, constructed according to Eq. (\ref{eq:variational_layer}), with an embedding overhead factor defined as in Definition \ref{definition:variational_circuit}. Then, if the overhead factor of the circuit is $\lambda=\omega(1)$, there exists a classical algorithm that can estimate the expectation value of the energy in the noisy quantum computation $\left<H\right>_{\mathrm{err}}=\mathrm{tr}(H\Phi_{\mathcal{C},p}^{\mathrm{noisy}}(\rho_0))$ up to absolute error $\epsilon n$ in time $T=O(\mathrm{poly}(n,1/\epsilon))$, $\forall p>p^*$, with $p^*=1-1/2^{O(1/\lambda)}=o(1)$. For a constant error rate $p=\Theta(1)$ the runtime is

\[
  T= (|E|+n) \log \left(\frac{nd}{\epsilon}\right)\left( \frac{nd}{\epsilon}\right)^{O(\frac{1}{\lambda})}.
\]
 
\end{lemma}

\begin{proof}

We write the Hamiltonian as $H=\sum_{k=1}^{g(n)} a_k O_k$, where the $O_k$ are Pauli operators with at most $2-$local terms. Here, $g(n)= n+|E|$, since the graph contains $n$ vertices and $|E|$ edges. We will first show that each $O_k$ can be approximated efficiently using the Pauli path method.

We denote the $m$ variational layers of the circuit by $L_1$,...,$L_m$, where each $L_j$ is built according to the prescription in Eq. (\ref{eq:variational_layer}), and contains layers of computing and permuting unitaries. That is, each $L_j$ can be written as $L_j=V_{j,1} S_{1} ... S_{M-1}V_{j,M} e^{-i \alpha_j H_M}$, where $M=O(1)$. We denote $c_j=\sum_{k=1}^M \mathrm{depth}(V_{j,k})+\mathrm{depth}\left(e^{-i \alpha_j H_M}\right)=O(1)$. That is, for a given variational layer $L_j$, $c_j$ denotes the total depth of the computing unitaries, while $\lambda$ denotes the total depth of the permuting unitaries. We denote by $t_j$ the time step when the first layer of gates of $L_j$ is applied, and by $t_j^{\prime}=t_{j}+\mathrm{depth}(L_j)$ the last one. 

A Pauli path is given by $d+1$ Pauli strings, $s=(s_0,s_1,...,s_d)$, with $s_k \in \{I,X,Y,Z\}^{\otimes n}$ representing the Pauli string at time $k$. We will follow the notation in the proof of Lemma \ref{lemma:counting_geometric}, and denote by $T_k(s_k)$ the number of $T$ gates in the layer $k$ of the quantum circuit that are applied on a position $(x,y)$ such that $s_k(x,y) \neq I$. That is $T_k(s_k)$ counts the number of $T$ gates that the Pauli path $s$ 'encounters' at time $k$, where we define encountering as applying the $T$ gate to a Pauli operator, but not to the identity.

We now bound:

\begin{align}\label{eq:bound_tk}
\sum_{k=t_j}^{t^{\prime}_j} |T_k(s_k)| \leq c_j \max_{t_i \leq k \leq t^{\prime}_i} \left|s_k\right|.
\end{align}

\noindent Here we used that for the layers that contain only SWAP gates $T_k(s_k)=0$, and therefore only the $V_{j,k}$ layers contribute to the sum. Furthermore, we can bound

\begin{align}\label{eq:bound_sum_sk}
\sum_{k=t_j}^{t^{\prime}_j} |s_k|\geq \lambda \min_{t_j \leq k \leq t^{\prime}_j} \left|s_k\right|,
\end{align}
\noindent since we know there are $\lambda$ layers that contain only SWAP gates. Finally, we note that

\begin{align}\label{eq:bound_min_sk}
\min_{t_j \leq k \leq t^{\prime}_j} \left|s_k\right| \geq \frac{1}{2^{c_j}}  \max_{t_j \leq k \leq t^{\prime}_j},
\end{align}
since the permuting layers  $S_{k}$ cannot decrease the weight of a Pauli string, while the computing layers $V_{j,k}$ can decrease it by at most a factor of $2$. 

Applying the bounds in Eqs.(\ref{eq:bound_tk}), (\ref{eq:bound_sum_sk}) and (\ref{eq:bound_min_sk}) leads to:

\[
\sum_{k=t_j}^{t_j^{\prime}} \left|T_k(s_k)\right| \leq \frac{c_j2^{c_j}}{\lambda} \sum_{k=t_j}^{t_j^{\prime}} \left|s_k\right|,
\]
\noindent with $c_j=O(1)$. Since this is expression true for all $L_j$, one finds that for all nonzero Pauli paths of weight $w$:

\[
\sum_{k=1}^d \left|T_k(s_k)\right| \leq \frac{\max_j\left(c_j2^{c_j}\right)}{\lambda}w =  Q w,
\]

\noindent where $Q=\max_j\left(c_j2^{c_j}\right)/\lambda=O(1/\lambda)$. Then, applying the technique in the proof of Lemma \ref{lemma:counting_geometric}, one can bound the number of nonzero Pauli paths of weight $w$, $N_w \leq 2^{Qw}$. The approximation error of the classical simulation is then (see Eq.(\ref{eq:F_w}) for the definition of $F_w$):

\begin{align}\label{eq:error_variational}
\left| \sum_{w>\ell} F_w (1-p)^w \right| \leq  \left| \sum_{w=\ell}^{nd} N_w (1-p)^w \right| \leq nd\left[2^Q(1-p)\right]^\ell \leq \epsilon.
\end{align}

\noindent Therefore, the error will decrease exponentially with $\ell$ if $2^Q(1-p)<1$. This condition will be fulfilled for $\forall p>p^*$, with

\[
p^*=1-1/2^{Q}=1-1/2^{O(1/\lambda)}=o(1).
\]
\noindent As a consequence, the approximation error of the algorithm will be small for every constant error rate $p$. 

To fulfill Eq.(\ref{eq:error_variational}) it suffices to pick a cut-off

\begin{align}\label{eq:exp:cutoff}
\ell=\frac{\log (nd/\epsilon)}{\log \left(1/(2^{Q}(1-p))\right)}.
\end{align}

\noindent The runtime is therefore $T=O(\ell \times 2^{Q \ell})=O(\mathrm{poly}(n,1/\epsilon))$, when $p>p^*$.

Furthermore, if one picks a constant error rate $p=\Theta(1)$, using Eq. (\ref{eq:exp:cutoff}), the cut-off scales as $\ell=O(\log (nd/\epsilon))$, since $Q=O(1/\lambda)=o(1)$.
 Therefore, the total runtime is 
 \[
 T=\ell \times 2^{Q\ell}= \log\left(\frac{nd}{\epsilon}\right) \left(\frac{nd}{\epsilon}\right)^{O\left(\frac{1}{\lambda}\right)}, \forall p=\Theta(1).
 \]
 \noindent Note that the runtine is subpolynomial.

Now, we bound the total error for the energy. $H$ is given by a linear combination of $g(n)=n+|E|$ Pauli observables $O_k$, $H=\sum_{k=1}^{g(n)}a_k O_k$. Using the Pauli path method above, we have shown that one can approximate each of the $O_k$ individually up to precision $\epsilon^{\prime}$. That is, the classical algorithm can output $\left<O_k\right>_{\ell}$ such that $\left|\left<O_k\right>_{\ell}-\left<O_k\right>_{\mathrm{err}}\right| \leq \epsilon^{\prime}$. Then, the output energy provided by the Pauli paths method is $\left<H\right>_{\ell}=\sum_k a_k \left<O_k\right>_{\ell}$. 



\noindent And the total error can be upper bounded by
\[
\left| \left<H\right>_{\mathrm{err}}-\left<H\right>_{\ell}\right| \leq \sum_{k=1}^{g(n)}|a_k| \epsilon^{\prime} \leq \epsilon^{\prime} (J_{\max} \Delta \frac{n}{2} + b_{\max}n),
\]

\noindent One can then pick $\epsilon^{\prime}=\epsilon/(J_{\max}\Delta/2+b_{\max})=\Theta(\epsilon)$, yielding a runtime:

\[
  T= (|E|+n) \log \left(\frac{nd}{\epsilon}\right)\left( \frac{nd}{\epsilon}\right)^{O(\frac{1}{\lambda})}.
\]

\end{proof}
\noindent Putting together Lemmas \ref{lemma:easy_instances} and \ref{lemma:Ising_hamiltonian_hardness} provides our result:

\begin{theorem}[Variational circuits]\label{theorem:variational_circuits}
Consider a target Ising Hamiltonian $H$ of the form of Eq. \ref{eq:ising_hamiltonian}, defined on a graph $G=(n,E)$ of bounded degree $\Delta=\Theta(1)$, and with couplings of similar strength, $\max_{i,j} |J_{ij}|/\min_{i,j} |J_{ij}| =O(1)$. Consider also a variational circuit that attempts to approximate the ground energy of $H$, constructed according to Eq. \ref{eq:variational_layer} and with an embedding overhead factor defined as in Definition \ref{definition:variational_circuit}. Then, there exists either:
\begin{itemize}
\item A classical polynomial time approximation algorithm that can approximate $E_0$, the ground energy of $H$, up to absolute error $\epsilon n $ in time $T=O\left(n2^{1/\epsilon^2}\right)$.

\item A classical algorithm that can approximate the expectation value of the energy of the noisy quantum circuit $\left<H\right>=\mathrm{tr}(H\Phi_{\mathcal{C},p}^{\mathrm{noisy}}(\rho_0))$ up to absolute error $\epsilon n$ in time $T=O(\mathrm{poly}(n,1/\epsilon))$, for any constant noise rate $p$.
\end{itemize}
\end{theorem}
\begin{proof}
The theorem follows directly from Lemmas \ref{lemma:easy_instances} and \ref{lemma:Ising_hamiltonian_hardness}. If $\lambda=O(1)$, Lemma \ref{lemma:easy_instances} provides a classical approximation algorithm that approximates the ground energy $E_0$ up to absolute error $\epsilon n$ in time $T=O\left(n2^{1/\epsilon^2}\right)$. On the other hand, if $\lambda=\omega(1)$, Lemma \ref{lemma:Ising_hamiltonian_hardness} shows that there is a classical efficient algorithm to approximate the output energy of the noisy quantum circuit up to absolute error $\epsilon n$ in time $T=O(\mathrm{poly}(n,1/\epsilon))$.
\end{proof}

We remark that these results have been obtained for circuits that are comprised of Clifford + $T$ gates because adding $T$ gates to the Clifford gate set is a natural way to obtain a universal gate set. However, these results can be trivially generalized to gate-sets that contain other non-Clifford gates, with the only price of obtaining worse constant factors in the expressions.

\section{Worst-case performance}\label{Section: counterexample}

In this section we show that the Pauli paths method fails to approximate noisy expectation values when naively applied. To do this, we we consider the expectation value of an observable $O$, which is constructed as a convex combination of $O(\mathrm{poly}(n))$ Pauli observables. This allows us to construct a circuit in which the error grows exponentially with the cut-off $\ell$. The construction of such a circuit $\mathcal{C}$ is remarkably simple, consisting only on a single layer of 3-qubit gates. Therefore, we remark that, even though the naive application of the Pauli path based method fails, the circuit is trivially easy to simulate classically.
In what follows, we will show how to construct the circuit $\mathcal{C}$. We assume, without loss of generality, that the system size $n$ is a multiple of $3$. Then, $\mathcal{C}$ can be constructed by simply applying a layer of $n/3$ $3$-qubit gates. The description of the $3$-qubit gate, which we denote by $V$, and $\mathcal{C}$, is sketched in Fig. \ref{fig:counterexample_circuit}, and detailed in the following lemma:

\begin{figure*}
    \centering
    \includegraphics[scale=0.8]{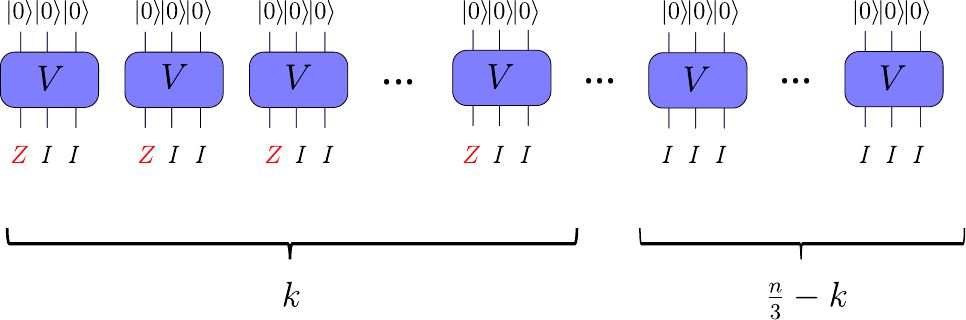}
    \caption{Sketch of the circuit $\mathcal{C}$ that we study. Without loss of generality, we assume that the system size $n$ is a multiple of $3$. The circuits of applying a layer of $3$-qubit unitaries to all the qubits, with initial state $\left(\ket{0}\bra{0}\right)^{\otimes n}$. We have also represented the observable $O_k=\prod_{i=1}^kZ_{3i-2}$, since the observable that we will measure is a convex combination of $O_k$ for different values of $k$.}
\label{fig:counterexample_circuit}
\end{figure*}

\begin{lemma} [$V$ gate]\label{lemma:V_gate}
Consider the inital state $\rho_0 = \left(\ket{0}\bra{0}\right)^{\otimes n}$ on $n$ qubits and the observable $O_k=\prod_{i=1}^k{Z_{3i-2}}$, with $k < n/3$. Then, one can construct a circuit $\mathcal{C}$ made of a single layer of 3-qubit unitaries such that $F_{2k}=(3/2)^k$.
\end{lemma}
\begin{proof}
We consider the 3-qubit unitary that implements majority voting in the following way:

\begin{multline}
\\
V\ket{000}=\ket{000} \quad , \quad V\ket{100}=\ket{011} \\
V\ket{001}=\ket{001} \quad , \quad V\ket{101}=\ket{101} \\
V\ket{010}=\ket{010} \quad , \quad V\ket{110}=\ket{110} \\
V\ket{011}=\ket{100} \quad , \quad V\ket{111}=\ket{111}. \\
\end{multline}

\noindent One can check that $V^{\dagger}(Z_1)V = (Z_1 + Z_2+Z_3)/2 - Z_1Z_2Z_3/2$.

Then, applying directly Eq.\ref{eq:F_w} with $\rho_0 = \ket{000} \bra{000}$, we find that the only nonzero Pauli paths are of the form $s=(s_0,s_1)$, with $s_1=Z_1$, and $s_0 \in \{Z_1,Z_2,Z_3,Z_1Z_2Z_3\}$. Therefore, we obtain:
\[
F_2=\sum_{|s|=2} f(s)=\sum f(Z_1,Z_1)+f(Z_2,Z_1)+f(Z_3,Z_1)=3/2,
\]

\noindent and $F_4=\sum_{|s|=4} f(s)=\sum f(Z_1Z_2Z_3,Z_1)=-1/2,$ while $F_3=F_1=0$.

Now, We consider the initial state $\rho_0 = \left(\ket{0} \bra{0} \right)^{\otimes n}$ and the observable $O_k=\prod_{i=1}^k{Z_{3i-2}}$. Assuming, without loss of generality, that $n$ is a multiple of $3$, we consider the circuit $\mathcal{C}$ that results from applying a layer of $n/3$ $V$ gates, which is simply $V^{\otimes n/3}=V_{123} \otimes V_{456} \otimes ... \otimes V_{n-2,n-1,n}$, where $V_{i,j,k}$ represents a $V$ gates applied on qubits $i,j,k$. We then compute:

\begin{multline}
(V^{\dagger})^{\otimes n/3}O_k V^{\otimes n/3}=\left[\bigotimes_{i=1}^kV^{\dagger}_{3i-2,3i-1,3i}Z_{3i-2}V_{3i-2,3i-1,3i}\right]\otimes I^{\otimes n-3k}= \\
=\left[\bigotimes_{i=1}^k\frac{1}{2}\left(Z_{3i-2}+Z_{3i-1}+Z_{3i}-Z_{3i-2}Z_{3i-1}Z_{3i}\right) \right]\otimes I^{\otimes n-3k}.
\end{multline}

\noindent This product structure makes it easy to compute the coefficient $F_w$ by decomposing it into $k$ blocks of $3$ qubits each, since for each of the blocks we have already computed the coefficients $F_1,F_2,F_3,F_4$:

\[
F_w=\sum_{i_1+...+i_k=w}\left(F_{i_1}+...+F_{i_k}\right),
\]
with $i_1,...,i_k \in \{1,2,3,4\}$, and $F_1=F_3=0$, $F_2=3/2$, $F_4=-1/2$.

One can then check that $F_{2k}=(3/2)^k$. Computing all the terms, one obtains

\begin{align}\label{eq:Fw}
F_w=\begin{cases}
    \left(3/2\right)^{2k-w/2} \left(-1/2\right)^{w/2-k}{k \choose 2k-w/2} & \emph{for } w \emph{ even, } w \geq 2k\\
    0 & \emph{otherwise}
\end{cases}
\end{align}

\end{proof}

We will now study some properties of the circuit $\mathcal{C}$, which we state in the follwing Lemma:

\begin{lemma}[Properties of $\mathcal{C}$] \label{lemma-Properties} 
Consider the coefficient $F_w$ when the circuit $\mathcal{C}$ is applied on the observable $O_k=\prod_{i=1}^k{Z_{3i-2}}$, denoted by $F_w(\mathcal{C}_k)$, and the error after truncating the Pauli path method with a cut-off $\ell$, denoted by $E_{\mathcal{C}_k}^{\ell}$:
\[
E^{(\ell)}_{\mathcal{C}_k}=\sum_{w>\ell} F_w(\mathcal{C}_k)(1-p)^w.
\]
Then, the following properties hold:

\begin{enumerate}

\item For $w<2k$ it is $F_w(\mathcal{C}_k)=0$.

\item For $\ell<2k$ it is $E^{(\ell)}_{\mathcal{C}_k}=(3(1-p)^2/2 - (1-p)^4/2 )^k=(1-p)^{\Omega(\ell)}$.

\item For $\ell>4k$ it is $E^{(\ell)}_{\mathcal{C}_k}=0$.

\item For $2k < \ell \leq 4k$ the sign of $E^{(\ell)}_{\mathcal{C}_k}$ alternates:

\[
\mathrm{sgn}(E^{(\ell)}_{\mathcal{C}_k})=(-1)^{\frac{\ell - 2k}{2}}.
\]
\item For $2k < \ell \leq 9k/4$:
\[
\left|E^{(\ell)}_{\mathcal{C}_k}\right| \geq \left(\frac{3}{2}(1-p)^2\right)^k-1=2^{\Omega(\ell)}, \forall p<1-\sqrt{2/3}.
\]

\end{enumerate}
\end{lemma}
\begin{proof}
The expression for $F_w(\mathcal{C}_k)$ is given by Eq. \ref{eq:Fw}, which already proves property 1 and 3. The term $E^{(\ell)}_{\mathcal{C}_k}$ can be obtained by performing a sum over the expression of $F_w(\mathcal{C}_k)$, which can be expressed in terms of the hypergeometric function ${}_2 F_1(a,b;c;z)$ (assuming, without loss of generality, that $k$ and $\ell$ are even):

\begin{multline}\label{eq:partial_sum}
E^{(\ell)}_{\mathcal{C}_k}=-(-1)^{\frac{\ell-2k}{2}} \left(2^{-k}\right)  \left(3^{2k - \frac{\ell}{2}-1} \right)  (1-p)^{\ell+2} {k \choose 2k-\ell/2-1} \times \\ \times {}_2 F_1(1,-2k + \ell/2+1;2-k+\ell/2;(1-p)^2/3).
\end{multline}

\noindent Property 2 can be verified by checking this formula when $\ell=2k-2$. Property 4 follows from the fact that all the terms in the formula above are positive except for $-(-1)^{\frac{\ell-k}{2}}$, which alternates between positive and negative when $\ell$ is increased. The hypergeometric function can be checked to be positive by the transformation

\begin{multline}\label{eq:hypergeometric}
{}_2 F_1(1,-2k + \ell/2+1;2-k+\ell/2;(1-p)^2/3) = \\
= \left[1-\frac{(1-p)^2}{3}\right]^k {}_2 F_1(-k+\ell/2+1,1+k;2-k+\ell/2;(1-p)^2/3),
\end{multline}
which is a sum of positive terms when $2k < \ell \leq 4k$.

Finally, property 5 can be verified by inspection of Eq. (\ref{eq:hypergeometric}): the term $|{}_2 F_1(-k+\ell/2+1,1+k;2-k+\ell/2;(1-p)^2/3)|$ can be checked to be monotonically increasing with $\ell$. As a consequence, using Eq. (\ref{eq:partial_sum}), $|E^{(\ell)}_{\mathcal{C}_k}|$ can be checked to be monotonically increasing with $\ell$, when $2k < \ell \leq 9k/4$ and $p<1-\sqrt{2/3}$. Finally, by direct application of the formula, one finds that $\left|E_{\mathcal{C}_k}^{(\ell)}\right| \geq \left|E_{\mathcal{C}_k}^{(\ell=2k-2)}\right| \geq \left(3(1-p)^2/2\right)^k - 1=2^{\Omega(\ell)}$, when $p<1-\sqrt{2/3}$ which completes the proof.

\end{proof}

With these ingredients, we are ready to construct a circuit which exhibits a large error when the cut-off is $\ell=\Theta(\log n)$:

\begin{lemma}\label{lemma:counterexample} Consider the initial product state $\rho_0 = \left(\ket{0} \bra{0} \right)^{\otimes n}$. Then, there exists a quantum circuit $\mathcal{C}$, and an observable $O$, expressed as a linear combination of $g(n)=\Theta\left((\log n)^2\right)$ Pauli observables, $O=\frac{1}{g(n)}\sum_{k=1}^{g(n)}O_{4k}$, such that the error incurred by the Pauli path based classical algorithm scales exponentially with the cut-off $\ell$, when $\ell=O(\log n)$:

\[
\left| \sum_{w>\ell} F_w (1-p)^w \right| = 2^{\Omega(\ell)}/\left(\log n\right)^2, \forall p<1-\sqrt{2/3}.
\]

\end{lemma}
\begin{proof}

We will consider the circuit $\mathcal{C}$ from Lemma \ref{lemma:V_gate}, the initial state $\rho_0=\left(\ket{0}\bra{0}\right)^{\otimes n}$, and the observable 
\[
O=\frac{1}{g(n)}\sum_{k=1}^{g(n)}O_{4k},
\]
for some function $g(n)=\Theta\left((\log n)^2\right)$. Here,
$O_{4k}=\prod_{i=1}^{4k}{Z_{3i-2}}$.

We denote the error by truncating the Pauli path method at a cut-off $\ell$ by

\[
E_{\mathcal{C}}^{(\ell)}=\sum_{w>\ell}F_w(1-p)^w,
\]

\noindent where we denote by $F_w$ the coefficient $F_w$ obtained for the circuit $\mathcal{C}$ and the initial observable $O$. We will also denote by $F_w(\mathcal{C}_{4k})$ the coefficient $F_w$ obtained for the circuit $\mathcal{C}$ and the initial observable $O_{4k}$. Furthermore, since $O$ can be written as a convex combination of the Pauli strings $O_{4k}$, one gets:

\[
F_w=\frac{1}{g(n)}\sum_{k=1}^{g(n)}F_w(\mathcal{C}_{4k}),
\]

\noindent and 

\[
E_{\mathcal{C}}^{(\ell)}=\frac{1}{g(n)}\sum_k E_{\mathcal{C}_{4k}}^{(\ell)}.
\]



\noindent From properties 2 and 3 of Lemma \ref{lemma-Properties}, we deduce that:

\begin{align}
\left|E_{\mathcal{C}_{4k}}^{(\ell)}\right|=\begin{cases}
    (1-p)^{\Omega(\ell)} & \emph{ if } \ell <  8k\\
    0 & \emph{ if } \ell >  16k.
    \end{cases}
\end{align}

\noindent We denote by $\mathcal{A}_\ell$ the set of integers $k$ such that $8k < \ell< \min(16k,g(n))$, and by $\mathcal{B}_\ell$ the set of integers $k$ such that $8k < \ell< \min(9k,g(n))$. Then:

\[
\left|E_\mathcal{C}^{(\ell)}\right|=\left|\frac{1}{g(n)}\sum_{k \in \mathcal{A}_\ell} E_{\mathcal{C}_{4k}}^{(\ell)}\right| + (1-p)^{\Omega(\ell)}.
\]

\noindent Furthermore, we notice that, due to property 4 of Lemma \ref{lemma-Properties}:
\[
\mathrm{sgn}(E_{\mathcal{C}_{4k}}^{(\ell)})=\mathrm{sgn}(E_{\mathcal{C}_{4j}}^{(\ell)}) \forall k,j\in \mathcal{A}_\ell.
\]

\noindent This is due to property 4 of Lemma \ref{lemma-Properties}. Since $\mathcal{C}_{4k}$ and $\mathcal{C}_{4j}$ only differ by a multiple of $4$, the signs of their coefficients $F_w$ will coincide. Therefore, since $\mathcal{B}_\ell \subset \mathcal{A}_\ell$:

\[
\left|\frac{1}{g(n)}\sum_{k \in \mathcal{A}_\ell} E_{\mathcal{C}_{4k}}^{(\ell)} \right| \geq  \left| \frac{1}{g(n)}\sum_{k \in \mathcal{B}_\ell} E_{\mathcal{C}_{4k}}^{(\ell)}  \right| \geq \max_{k \in \mathcal{B}_\ell} \left| E_{\mathcal{C}_{4k}}^{(\ell)} \right|/g(n).
\]

This bound is only meaningful if the set $\mathcal{B}_\ell$ is not empty. However, one can check that, for any $\ell$, $\mathcal{B}_\ell \neq \emptyset $, since one can always find a $k \in \mathbb{N}$ such that $8k < \ell < 9k$. 
Applying property 5 of Lemma \ref{lemma-Properties} then yields:
\[
\max_{k \in \mathcal{B}_\ell} \left| E_{\mathcal{C}_{4k}}^{(\ell)} \right|=2^{\Omega(\ell)}/g(n), \forall p<1-\sqrt{2/3},
\]

\noindent and therefore:
\[
\left|E_\mathcal{C}^{(\ell)}\right| = 2^{\Omega(\ell)}/g(n),
\]
 which proves the lemma.
 \end{proof}

Finally, we can restate this Lemma to produce our no-go result: in general, a cut-off of $\ell=\Theta(\log n)$ is not enough to guarantee that the error in the simulation is small. As a consequence, since $\ell=\omega(\log n)$ is needed, the running time will be superpolynomial in general:

\begin{theorem}[No-go result] \label{theorem:counterexample}
Consider the Pauli paths method to compute the expectation value of a convex combination of $g(n)=O(\mathrm{poly}(n))$Pauli observables $O=\sum_{k=1}^{g(n)} a_k O_k$, applied with a cut-off $\ell=\Omega(\log n)$, and with the error given by
$\mathrm{Error}=\left|\left<O\right>_{\ell} - \left<O\right>_{\mathrm{err}} \right|$. Then, it is necessary that the cut-off scales as $\ell=\omega(\log n)$ to ensure that, for all error rates and for any $O(\log n)$ depth quantum circuit, there exists a constant $\epsilon=\Theta(1)$ such that $\mathrm{Error} \leq \epsilon$.
\end{theorem}
\begin{proof}
The proof follows directly from Lemma \ref{lemma:counterexample}. There, a quantum circuit $\mathcal{C}$ is constructed such that, for $\ell=\Theta(\log n)$ and $p<1-\sqrt{2/3}$, the error scales as $\mathrm{Error}=2^{\Omega(\ell)}/(\log n)^2$. Hence, if one wants to ensure that for small error rates the error is bounded by a constant $\mathrm{Error} \leq \epsilon$, one has compute more terms of the Pauli path sum, setting at least $\ell=\omega(\log n)$.
\end{proof}

\section{Conclusions}
We have studied a Pauli path based algorithm to compute expectation values of noisy quantum circuits. This algorithm leverages the fact that the Pauli paths with small weight can be computed efficiently on a classical computer, while the paths with large weight decay rapidly with noise. It is known that this method works for generic instances of noisy quantum circuit, but its worst-case performance remains unexplored. Using this method we have shown that, for circuits comprised of Clifford + $T$ gates, there is a threshold error rate beyond which the classical computation of local observables is efficient. The threshold error rate depends roughly on both the fraction of $T$ gates in the circuit and their geometric distribution. This result illustrates the ``competition'' between the quantum magic and the noise in the system: if noise is introduced at a greater rate than quantum magic, then classical simulation becomes possible in all cases. However, this result might not be practical in its current form, since it leads to stringent conditions on the $T$ gates, and a large exponential in the running time. It remains an interesting question whether this can be improved. As an application of these results, we study 2D QAOA circuits that attempt to find low-energy state of classical Ising Hamiltonians on bounded degree graphs. We find that, on hard instances, which correspond to the Ising model's graph being geometrically non-local, these circuits are classically easy to simulate for any constant noise rate, since they are dominated by SWAP gates.

We have also explored the limitations of the algorithm, and provide an explicit instance of the problem that is classically trivial to simulate, but where naive application of the algorithm fails nonetheless. This serves as a simple example of how the truncations in the algorithm can lead to exponentially growing errors.
Finally, we have initiated a study of the trade-off that exists between classical complexity and sensitivity to noise in a subset of quantum circuits. We believe that it remains an interesting open problem whether this can be generalized further: are there circuits that are both classically hard yet somewhat robust to noise in the NISQ setting?

\emph{Note:} Concurrently with the preparation of this manuscript, a new average-case analysis of the Pauli paths method appeared \cite{schuster2024polynomialtime}. There, the classical Pauli path algorithm was shown to approximate noisy expectation values of fixed quantum circuits with small average error for initial states randomly drawn from an ensemble, for any constant error rate. While their main result is an average-case analysis, they also analyze worst-case performance on the more restricted computational model of noisy $\mathrm{DQC}1$, showing classical simulability. Our work complements this result by performing a worst-case analysis on a different family of quantum circuits, comprised of fixed circuits and fixed initial (product) states.

\begin{acknowledgements}
The research is part of the Munich Quantum Valley, which is supported by the Bavarian State Government with funds from the High tech Agenda Bayern Plus. We acknowledge funding from the German Federal Ministry of Education and Research (BMBF) through EQUAHUMO (Grant No. 13N16066) within the funding program quantum technologies—from basic research to market.
\end{acknowledgements}

\bibliographystyle{quantum}
\bibliography{references.bib}

\appendix

 \section{Root analysis}\label{section:root_analysis}
In this section we analyze the relation between the stability to noise and the classical complexity of quantum circuits. While in the general case the relation is uncertain, we find a subclass of circuits that exhibit a trade-off: if the computation of expectation values of observables is classically hard, then the circuit is highly fragile to noise, and the noisy expectation value vanishes superpolynomially fast in the system size.

Recalling Eq.(\ref{eq:polynomial_O}), we write:
\[
 \left<O\right>_{\mathrm{err}}=\sum_w F_w(1-p)^w.
 \]

\noindent This is a polynomial in $p$, and we will analyze the fragility to noise of $ \left<O\right>_{\mathrm{err}}$ as a function of the roots. Specifically, if all the roots are real, there are only two possibilities: either $\left<O\right>_{\mathrm{err}}$ is classically easy to compute, or it is highly sensitive to noise.  However, we note that the assumption that all the roots are real is not general. Furthermore, since computing the roots is as hard as computing the expectation value, it is not easy to characterize the subset of circuits with real roots. However, it is trivial to provide circuits that fulfill this condition: for example, all Clifford circuits will have all roots in $p=1$.

In our analysis, we will denote 
\[
\|F\|_2=\sqrt{\sum_w F_w^2}.
\]

\noindent We will also denote $x=1-p$, and write the polynomial as:

\[
Q(x)=\sum_{w=d}^{M}F_wx^w,
\]

\noindent with $M$ being the degree of the polynomial.
The quantity $\|F\|_2$ is the basis of our analysis. We will show that, if $\|F\|_2=O(\mathrm{poly(n)})$, the classical simulation is always easy, while if $\|F\|_2=\omega(\mathrm{poly(n)})$, the expectation value is $1/\omega(\mathrm{poly}(n))$ small, for small constant error rates, which implies fragility to noise.

We begin by bounding the number of roots of $Q(x)$ that are within a constant distant of the origin, as a function of $\|F\|_2$:

 \begin{lemma}\label{lemma:roots_origin}
For circuits comprised of Clifford and $T$ gates, the polynomial $Q(x)$ contains $\Omega(\log \|F\|_2)$ roots within a distance $R=\Theta(1)$ of the origin.

 \end{lemma}
\begin{proof}
We will begin by bounding $M$, the degree of $Q(x)$. Using that $|F_w| \leq 2^w$ (from the counting argument of Lemma \ref{lemma:counting_geometric}), we can bound:

\[
\|F\|_2^2=\sum_{w=d}^M |F_w|^2 \leq \sum_{w=d}^M 4^w =2^{O(M)}.
\]
This implies that $M=\Omega\left(\log \|F\|_2\right)$.

We consider the coefficient $F_M$. We realize that $|F_M| \geq \sqrt{2}^{-M}$, since there are only Clifford and $T$ gates in the circuit. One can then use Corollary 9.1.2 of \cite{rahman2002analytic} to bound the $k$-th root of $Q(x)$:

\[
|r_k| \leq \left( \frac{\|F\|_2}{|F_M|} \right)^{\frac{1}{M-k-1}} \leq \left(\sqrt{2}^M \|F\|_2\right)^{\frac{1}{M-k-1}}.
\]

\noindent Picking $k=1+M/2$ this yields

\[
|r_{1+M/2}| \leq \left(\sqrt{2}\|F\|_2^{1/M}\right)^{2}=R.
\]

\noindent with $R=O(1)$ because $\|F\|_2^{1/M}=O(1)$, which proves the Lemma. 
\end{proof}

We continue by showing that, if $g(n)$ roots are within a constant distance of the radio, and all the roots are real, $Q(x)$ must decay exponentially with $g(n)$:

 \begin{lemma} \label{lemma:polynomial_and_roots}
 Consider a polynomial $Q(x)$ with real roots $r_1, ... , r_M$, with $|Q(x)| \leq 1 \forall x \in [0,1]$, a constant $R=\Theta(1)$, and a small constant $\epsilon \ll 1$. If there are at least $g(n)=\omega(\mathrm{poly}(n))$ roots within a distance $R$ of the origin, no roots in the interval $x \in [1-\epsilon,1)$, and $|Q(1)| = \Omega (\frac{1}{\mathrm{poly}(n)})$, then

 \[
 \left|Q(x)\right| = e^{-\Omega(g(n))x^2}, \forall x \in [1-\epsilon,1).
 \]
 \end{lemma}
\begin{proof}

We write $ Y(x)=\log |Q(x)|$. Taking the derivative, one can check that

 \[
 Y^{\prime}(x)=\frac{Q^{\prime}(x)}{Q(x)}=\sum_{i=1}^M \frac{1}{x-r_i},
 \]
 and
 \[
 Y^{\prime \prime}(x)=-\sum_{i=1}^M \frac{1}{(x-r_i)^2}.
 \]

\noindent This expression allows us to relate the behavior of $Q(x)$ to the number of roots that are within constant distance of the origin: 
 
One can directly bound

\[
Y^{\prime \prime}(x)=-\sum_{i=1}^M \frac{1}{(x-r_i)^2} \leq - \frac{g(n)}{(x+R)^2} \leq -c g(n) \forall x \in [1-\epsilon,1).
\]

\noindent Integrating this expression twice gives:

\[
 Y(x)-Y(1) \leq - Y^{\prime}(1) (1-x) - \frac{c}{2}(1-x)^2,
\]

\noindent and therefore

\[
|Q(x)| \leq |Q(1)| e^{-\frac{cg(n)}{2}(1-x)^2 -Y^{\prime}(1) (1-x)}.
\]

\noindent If $Y^{\prime}(1) \geq 0$, this expression proves the Lemma. If $Y^{\prime}(1) < 0$, we will consider the point $x_* \in [1-\epsilon,1)$ such that $Y^{\prime}(x_*)=0$. Again, integrating $Y^{\prime \prime}(x)$ twice yields

\[
\frac{|Q(x_*)|}{|Q(1)|} \geq e^{\frac{c}{2}g(n) (1-x_*)^2}.
\]
\noindent However, since $|Q(x_*)| \leq 1$, this places the restriction

\[
(1-x_*)^2 \leq \frac{2\log \frac{1}{|Q(1)|}}{c g(n)}=o(1),
\]
\noindent where we have used that $1/|Q(1)|=O(\mathrm{poly}(n))$ and $g(n)=\omega(\mathrm{poly}(n))$. One can then integrate starting from $x_*$:

\[
|Q(x)| \leq  |Q(x_*)| e^{-\frac{c}{2}g(n) (x-x_*)^2}, \forall x \in [1-\epsilon,x_*].
\]
\noindent Hence, the Lemma is satisfied, since $|Q(x_*)| \leq 1$, and $(x-x_*)^2=\Theta(1)$. 

\end{proof}

With these ingredients, we are now ready to state the main theorem:

\begin{theorem}\label{theorem:root_analysis}
Consider a circuit $\mathcal{C}$ comprised of Clifford + $T$ gates, and an observable $\left<O\right>_{\mathrm{err}}$, under the assumptions that all the roots of the polynomial $\left<O\right>_{\mathrm{err}}(p)$ are real, that there are no roots in the interval $(0,p^*]$ for some constant $p^*=\Theta(1)$, and that in the noiseless case $\left| \left<O\right>(p=0)\right| = \Omega(1/\mathrm{poly}(n))$. 
Then, if there is no classical algorithm that can approximate the expectation value $\left<O\right>_{\mathrm{err}}(p)$ to precision $\epsilon$ in time $T=O\left(\mathrm{poly}(n,1/\epsilon\right))$ for all constant noise rates $p\in\left(0,1\right]$, the noisy expectation value $\left|\left<O\right>_{\mathrm{err}}(p)\right|$ must be vanishingly small: 

\[
 \left|\left<O\right>_{\mathrm{err}}(p)\right| =\frac{1}{\omega(\mathrm{poly(n)})}, \forall p \in (0,p^*].
\]

\end{theorem}

\begin{proof} We first assume that $\|F\|_2 = O(\mathrm{poly}(n))$. The classical algorithm consists on computing all the nonzero Pauli paths up to a Pauli weight $\ell=O(\log n)$. The error is bounded using the Cauchy-Schwarz inequality:

\[
\left| \left<O\right>_{\mathrm{err}}-\left<O\right>_{\ell}\right|=\left| \sum_{w>\ell}F_w \left(1-p\right)^w \right| \ \leq (1-p)^{\ell} \sum_{w>\ell} |F_w| \leq (1-p)^{\ell} nd \|F\|_2.
\]

\noindent To ensure a precision $\epsilon$, it is enough to pick

\[
\ell = \frac{\log \left( \frac{\epsilon}{nd \|F\|_2} \right)}{\log (1-p)}.
\]

\noindent Therefore, $\ell=O(\log n)$ as long as $\|F\|_2= O(\mathrm{poly}(n))$, and thus the classical algorithm succeeds $\forall p \in [0,1)$ in time $T=\mathrm{poly}(n,1/\epsilon)$.

From this we deduce that $\|F\|_2 = \omega(\mathrm{poly}(n))$ is needed for the classical efficient algorithm to fail. However, in the following we will show that, if the roots of the polynomial are real, this implies fragility to noise. In Lemma \ref{lemma:roots_origin} we show that there must be $\Omega(\|F\|_2)$ roots within a distance $R=\Theta(1)$ of the origin, and in Lemma \ref{lemma:polynomial_and_roots} we show that this implies that $\left|\left<O\right>_{\mathrm{err}}(p)\right| = 1/ \Omega(\|F\|_2) \forall p \in (0,p^*]$ for some constant $p^*$.

\end{proof}

\section{Asymptotic notation}\label{appendix:notation}

Throughout the paper employ the following asymptotic notation commonly used in complexity theory \cite{cormen2022introduction}:

\begin{center}
\begin{tabular}{ |c|c|c| } 
\hline
 Notation & Formal definition & Informal description \\ 
 \hline
 $f(n)=o(g(n))$ & $\forall k>0 \exists n_0 \forall n>n_0: |f(n)| < kg(n)$ & $f(n)$ grows strictly slower than $g(n)$ \\  
 $f(n)=O(g(n))$ & $\exists k>0 \exists n_0 \forall n>n_0: |f(n)| \leq kg(n)$ & $f(n)$ grows no faster than $g(n)$  \\
 $f(n)=\Omega(g(n))$ & $\exists k>0 \exists n_0 \forall n>n_0: |f(n)| \geq kg(n)$ & $f(n)$ grows at least as fast as $g(n)$ \\ 
 $f(n)=\omega(g(n))$ & $\forall k>0 \exists n_0 \forall n>n_0: |f(n)| > kg(n)$ & $f(n)$ grows faster than $g(n)$ \\
 \hline
\end{tabular}
\end{center}

\end{document}